\documentclass[a4paper,USenglish,cleveref, autoref]{svjour3}

\title{Digraph Coloring and Distance to Acyclicity} 

\titlerunning{Digraph Coloring and Distance to Acyclicity}

\author{Ararat Harutyunyan \and Michael Lampis\thanks{Michael Lampis was partially supported by a grant from the French National Research Agency under the JCJC program (ASSK: ANR-18-CE40-0025-01).} \and Nikolaos Melissinos}

\authorrunning{A. Harutyunyan,  M.  Lampis,
and N. Melissinos}

\institute{ Universit\'e Paris-Dauphine, PSL Research University, CNRS, UMR 7243, LAMSADE, Paris, France \email{ararat.harutyunyan@dauphine.fr, michail.lampis@dauphine.fr, nikolaos.melissinos@dauphine.eu} }

\usepackage{todonotes}
\newcommand\tw{\mathrm{tw}}

\newcommand\td{\mathrm{td}}

\tikzstyle{vertex}=[circle, draw, inner sep=1.2pt, minimum width=4pt, minimum size=0.4cm]
\tikzstyle{vertex2}=[circle, draw, inner sep=0pt, minimum width=4pt, minimum size=0.15cm]
\usetikzlibrary{decorations,decorations.pathmorphing,decorations.pathreplacing,fit,positioning,arrows}

\usepackage{algpseudocode}
\usepackage{algorithm}

\begin{document}

\maketitle

\begin{abstract}

In $k$-\textsc{Digraph Coloring} we are given a digraph and are asked to
partition its vertices into at most $k$ sets, so that each set induces a DAG.
This well-known problem is NP-hard, as it generalizes (undirected)
$k$-\textsc{Coloring}, but becomes trivial if the input digraph is acyclic.
This poses the natural parameterized complexity question of what happens when the
input is ``almost'' acyclic. In this paper we study this question using
parameters that measure the input's distance to acyclicity in either the
directed or the undirected sense.

In the directed sense perhaps the most natural notion of distance to acyclicity
is directed feedback vertex set (DFVS). It is already known that, for all $k\ge
2$, $k$-\textsc{Digraph Coloring} is NP-hard on digraphs of DFVS at most $k+4$.
We strengthen this result to show that, for all $k\ge 2$, $k$-\textsc{Digraph
Coloring} is already NP-hard  for DFVS exactly $k$. This immediately provides a
dichotomy, as $k$-\textsc{Digraph Coloring} is trivial if DFVS is at most
$k-1$.  Refining our reduction we obtain three further consequences: (i) $2$-\textsc{Digraph Coloring} is NP-hard for oriented graphs of feedback
vertex set (FVS) at most $3$; (ii) for all
$k\ge 2$, $k$-\textsc{Digraph Coloring} is NP-hard for graphs of feedback
\emph{arc} set (FAS) at most $k^2$; interestingly, this leads to a second dichotomy,
as we show that the problem is FPT by $k$ if FAS is at most $k^2-1$; (iii)
$k$-\textsc{Digraph Coloring} is NP-hard for graphs of DFVS $k$, even if the
maximum degree $\Delta$ is at most $4k-1$; we show that this is also
\emph{almost} tight, as the problem becomes FPT for DFVS $k$ and $\Delta\le
4k-3$.

Since these results imply that the problem is also NP-hard on graphs of bounded directed treewidth, we then  consider parameters that
measure the distance from acyclicity of the underlying graph.  On the positive
side, we show that $k$-\textsc{Digraph Coloring} admits an FPT algorithm
parameterized by treewidth, whose parameter dependence is $(\tw!)k^{\tw}$.
Since this is considerably worse than the $k^{\tw}$ dependence of (undirected)
$k$-\textsc{Coloring}, we pose the question of whether the $\tw!$ factor can be
eliminated. Our main contribution in this part is to settle this question in
the negative and show that our algorithm is essentially optimal, even for the
much more restricted parameter treedepth and for $k=2$. Specifically, we show
that an FPT algorithm solving $2$-\textsc{Digraph Coloring} with dependence
$\td^{o(\td)}$ would contradict the ETH.

In the end, we consider the class of tournaments. It is known that deciding whether a tournament is $2$-colorable is NP-complete. We present an algorithm that decides if we can $2$-color a tournament in $O^*(\sqrt[3]{6}^n)$ time.


\CRclass{Mathematics of computing$\rightarrow$Graph algorithms \and Theory of Computation $\rightarrow$ Design and Analysis of Algorithms $\rightarrow$ Parameterized Complexity and Exact Algorithms}

\end{abstract}


\section{Introduction}

In  \textsc{Digraph Coloring}, we are given a digraph $D$ and are asked to
calculate the smallest $k$ such that the vertices of $D$ can be partitioned
into $k$ \emph{acyclic} sets. In other words, the objective of this problem is
to color the vertices with the minimum number of colors so that no directed
cycle is monochromatic.  This notion is called the \emph{dichromatic number} and it was introduced by V.
Neumann-Lara \cite{NL82}. More recently, digraph coloring has received much
attention, in part because it turns out that many results about the chromatic
number of undirected graphs quite naturally carry over to the dichromatic
number of digraphs
\cite{AboulkerCHLMT19,AndresH15,BensmailHL18,DBFJKM2004,CHZ2007,GurskiKR20,H11,HKMR,HarutyunyanLTW19,HochstattlerSS20,LiM17,MillaniSW19,Mohar03,SteinerW20}.
We note that \textsc{Digraph Coloring} generalizes \textsc{Coloring} (if we
simply replace all edges of a graph by pairs of anti-parallel arcs) and is
therefore NP-complete.

In this paper we are interested in the computational complexity of
\textsc{Digraph Coloring} from the point of view of structural parameterized
complexity\footnote{In the remainder, we assume the reader is familiar with the
basics of parameterized complexity theory, such as the class FPT, as given in
standard textbooks \cite{CyganFKLMPPS15}.}. Our main motivation for studying
this is that (undirected) \textsc{Coloring} is a problem of central importance
in this area whose complexity is well-understood, and it is natural to hope
that some of the known tractability results may carry over to digraphs --
especially because, as we mentioned, \textsc{Digraph Coloring} seems to behave
as a very close counterpart to \textsc{Coloring} in many respects.  In
particular, for undirected graphs, the complexity of \textsc{Coloring} for
``almost-acyclic'' graphs is very precisely known: for all $k\ge 3$ there is a
$O^*(k^{\tw})$ algorithm, where $\tw$ is the input graph's treewidth, and this
is optimal (under the SETH) even if we replace treewidth by much more
restrictive parameters \cite{JaffkeJ17,LokshtanovMS18}. Can we achieve the same
amount of precision for \textsc{Digraph Coloring}?

\paragraph{Our results:} The main question motivating this paper is therefore
the following: Does \textsc{Digraph Coloring} also become tractable for
``almost-acyclic'' inputs? We attack this question from two directions.

First, in Section~\ref{sec:bounded:fvs}, we consider the notion of acyclicity in the digraph sense and study
cases where the input digraph is close to being a DAG. Possibly the most
natural such measure is directed feedback vertex set (DFVS), which is the
minimum number of vertices whose removal destroys all directed cycles.  The
problem is paraNP-hard for this parameter, as for all fixed $k\ge 2$,
$k$-\textsc{Digraph Coloring} is already known to be NP-hard,  for inputs of
DFVS at most $k+4$ \cite{MillaniSW19}.  Our first contribution is to tighten
this result by showing that actually $k$-\textsc{Digraph Coloring} is already
NP-hard for DFVS of size \emph{exactly} $k$.  This closes the gap left by the
reduction of \cite{MillaniSW19} and provides a complete dichotomy, as the
problem is trivially FPT by $k$ when the DFVS has size strictly smaller than
$k$ (the only non-trivial part of the problem in this case is to find the DFVS
\cite{ChenLLOR08}).  In the end of this section we consider $2$-\textsc{Digraph Coloring} on oriented graphs. We prove that it is NP-hard to decide if an oriented graph is $2$-colorable even in cases where the size of DFVS is $3$. This is tight as there exists an easy argument showing that all oriented graphs with DFVS $k$ are $k$-colorable.

In Section~\ref{sec:bounded:fas:degree} we investigate if by considering a more restricted notion of near-acyclicity, or by imposing further restrictions, such as bounding the
maximum degree of the graph, could lead to an FPT algorithm.
Unfortunately, we show that neither of these
suffices to make the problem tractable. In particular, by refining our
reduction we obtain the following: First, we show that for all $k\ge 2$,
$k$-\textsc{Digraph Coloring} is NP-hard for digraphs of feedback \emph{arc}
set (FAS) $k^2$, that is, digraphs where there exists a set of $k^2$ arcs whose
removal destroys all cycles (feedback arc set is of course a more restrictive
parameter than feedback vertex set).  Interestingly, this also leads us to a
complete dichotomy, this time for the parameter FAS: we show that $k$-coloring
becomes FPT (by $k$) on graphs of FAS at most $k^2-1$, by an argument that
reduces this problem to coloring a subdigraph with at most $O(k^2)$ vertices,
and hence the correct complexity threshold for this parameter is $k^2$.
Second, we show that $k$-coloring a digraph with DFVS $k$  remains NP-hard even
if the maximum degree is at most $4k-1$. This further strengthens the reduction
of \cite{MillaniSW19}, which showed that the problem is NP-hard for bounded
\emph{degeneracy} (rather than degree). Almost completing the picture, we show
that $k$-coloring a digraph with DFVS $k$ and maximum degree at most $4k-3$ is
FPT by $k$, leaving open only the case where the DFVS is exactly $k$ and the
maximum degree exactly $4k-2$.

In Section~\ref{sec:bounded:treewidth}, because of the negative results for DFVS and FAS, we deiced to consider 
as parameter the treewidth
of the underlying graph.  It turns out that, finally, this suffices to lead to
an FPT algorithm, obtained with standard DP techniques.  However, our algorithm
has a somewhat disappointing running time of $(\tw!)k^{\tw}n^{O(1)}$, which is
significantly worse than the $k^{\tw}n^{O(1)}$ complexity which is known to be
optimal for undirected \textsc{Coloring}, especially for small values of $k$.
This raises the question of whether the extra $(\tw!)$ factor can be removed.
Our main contribution in this part is to show that this is likely impossible,
even for a more restricted case.  Specifically, we show that if the ETH is
true, no algorithm can solve $2$-\textsc{Digraph Coloring} in time
$td^{o(\td)}n^{O(1)}$, where $\td$ is the input graph's treedepth, a parameter
more restrictive than treewidth (and pathwidth). As a result, this paper makes
a counterpoint to the line of research that seeks to find ways in which
dichromatic number replicates the behavior of chromatic number in the realm of
digraphs by pinpointing one important aspect where the two notions are quite
different, namely their complexity with respect to treewidth.

Finally, in Section~\ref{sec:2-col:tournament}, we consider tournaments. It is already known that $2$-\textsc{Digraph Coloring} is NP-hard for tournaments~\cite{CHZ2007}.
The exhaustive algorithm to check if a tournament is $2$-colorable takes $O^*(2^n)$ time as there exists $2^n$ possible $2$-colorings for a graph. We improve this running time by proposing an algorithm that answers the same question in $O^*(\sqrt[3]{6}^n)$.

\paragraph{Other related work:} Structural parameterizations of
\textsc{Digraph Coloring} have been studied in \cite{SteinerW20}, who showed
that the problem is FPT by modular width generalizing the algorithms of
\cite{GajarskyLO13,Lampis12}; and \cite{GurskiKR20} who showed that the problem
is in XP by clique-width (note that hardness results for \textsc{Coloring} rule
out an FPT algorithm in this case \cite{FominGLS10,FominGLSZ19,Lampis18}). Our
results on the hardness of the problem for bounded DFVS and FAS build upon the
work of \cite{MillaniSW19}. The fact that the problem is hard for bounded DFVS
implies that it is also hard for most versions of directed treewidth, including
DAG-width, Kelly-width, and directed pathwidth
\cite{BerwangerDHKO12,GanianHK0ORS16,HunterK08,JohnsonRST01,LampisKM11}.
Indeed, hardness for FAS implies also hardness for bounded elimination width, a
more recently introduced restriction of directed treewidth \cite{FernauM14}.
For undirected treewidth, a problem with similar behavior is DFVS: (undirected)
FVS is solvable in $O^*(3^{\tw})$ \cite{CyganNPPRW11} but DFVS cannot be solved
in time $\tw^{o(\tw)}n^{O(1)}$, and this is tight under the ETH
\cite{BonamyKNPSW18}.  For other natural problems whose complexity by treewidth
is $tw^{\Theta(\tw)}$ see \cite{BasteST20,abs-2007-14179,BonnetBKM19}

With respect to maximum degree, it is not hard to see that $k$-\textsc{Digraph
Coloring} is NP-hard for graphs of maximum degree $2k+2$, because
$k$-\textsc{Coloring} is NP-hard for graphs of maximum degree $k+1$, for all
$k\ge 3$ \footnote{Note that this argument does not prove that
$2$-\textsc{Digraph Coloring} is NP-hard for maximum degree $6$, but this is
not too hard to show. We give a proof in Theorem \ref{thm:easy} for the sake of
completeness.}. On the converse side, using a generalization of Brooks' theorem
due to Mohar \cite{Mohar10} one can see that $k$-\textsc{Digraph Coloring}
digraphs of maximum degree $2k$ is in P. 
This leaves as the only open case digraphs of degree $2k+1$, which in a sense
mirrors our results for digraphs of DFVS $k$ and degree $4k-2$.  We note that
the NP-hardness of $2$-\textsc{Digraph Coloring} for bounded degree graphs is
known even for graphs of large girth, but the degree bound follows the imposed
bound on the girth \cite{FederHellSubi}.

\section{Definitions, Notation and Preliminaries}

We use standard graph-theoretic notation.  All digraphs are loopless and have
no parallel arcs; two oppositely oriented arcs between the same pair of
vertices, however, are allowed and are called a \emph{digon}.  Oriented graphs are digraphs which do not contain any digons. The in-degree
(respectively, out-degree) of a vertex is the number of arcs coming into
(respectively going out of) a vertex. The degree of a vertex is the sum of its
in-degree and out-degree. For a set of arcs $F$, $V(F)$ denotes the set of
their endpoints. For a set of vertices $S$ of a digraph $D$, $D[S]$ denotes the
digraph induced by $S$ and $N[S]$ denotes the closed neighborhood of $S$, that
is, $S$ and all vertices that have an arc to or from $S$.  

The \emph{chromatic number} of a graph $G$ is the minimum number of colors $k$
needed to color the vertices of $G$ such that each color class is an
independent set. We say that a digraph $D=(V,E)$ is \emph{$k$-colorable} if we
can color the vertices of $D$ with $k$ colors such that each color class
induces an acyclic subdigraph (such a coloring is called a \emph{proper
$k$-coloring}).  The \emph{dichromatic number}, denoted by $\vec{\chi}(D)$, is
the minimum number $k$ for which $D$ is $k$-colorable. The maximum degree of a graph or digraph is denoted with $\Delta$. 

A subset of vertices $S \subset V$ of $D$ is called a \emph{feedback vertex
set} if $D-S$ is acyclic. 

\begin{remark}\label{remark:dfvs:k-1:k:colorable} Every digraph $D=(V,E)$ with feedback vertex set of size at most
$k-1$ is $k$-colorable.  \end{remark}

The remark holds because we can use distinct colors for the vertices of the
feedback vertex set and the remaining color for the rest of the graph. 

A subset of arcs $A \subset E$ of $D$ is called a
\emph{feedback arc set} if $D-A$ is acyclic.  For the definition of treewidth
and nice tree decompositions we refer the reader to \cite{CyganFKLMPPS15}. A
graph $G$ has treedepth at most $k$ if one of the following holds: (i) $G$ has
at most $k$ vertices (ii) $G$ is disconnected and all its components have
treedepth at most $k$ (iii) there exists $u\in V(G)$ such that $G-u$ has
treedepth at most $k-1$. We use $\tw(G), \td(G)$ to denote the treewidth and
treedepth of a graph. It is known that $\tw(G)\le \td(G)$ for all graphs $G$.


The Exponential Time Hypothesis (ETH) \cite{ImpagliazzoPZ01} states that there
is a constant $c>1$ such that no algorithm which decides if \textsc{3-SAT}
formulas with $n$ variables and $m$ clauses are satisfiable can run in time
$c^{n+m}$. In this paper we will use the simpler (and slightly weaker) version
of the ETH which simply states that \textsc{3-SAT} cannot be solved in time
$2^{o(n+m)}$.

Throughout the paper, when $n$ is a positive integer we use $[n]$ to denote the
set $\{1,\ldots,n\}$. For a set $V$ an ordering of $V$ is an injective function
$\sigma:V\to [ |V| ]$. It is a well-known fact that a digraph $D$ is acyclic if
and only if there exists an ordering $\sigma$ of $V(D)$ such that for all arcs
$uv$ we have $\sigma(u)<\sigma(v)$. This is called a topological ordering of
$D$.

 We conclude this section with a preliminary theorem. As we mentioned, the argument from (undirected) graph coloring that shows why $k$-\textsc{Digraph Coloring} is NP-hard for digraphs of $\Delta = 2k+ 2$ does not hold for $k=2$. Our first theorem gives a proof for this case.

\begin{theorem}\label{thm:easy} It is NP-hard to decide if a given digraph with
maximum degree $6$ is $2$-colorable.  \end{theorem}

\begin{proof}

We perform a reduction from \textsc{NAE-3-SAT}, a variant of \textsc{3-SAT}
where we are asked to find an assignment that sets at least one literal to True
and one to False in each clause. First we remark that this problem remains
NP-hard if all literals appear at most twice. 

To see this, suppose that $x$
appears $\ell\ge 4$ times in $\phi$. We replace each appearance of $x$ with a
fresh variable $x_i$, $i\in [\ell]$ and add to the formula the clauses $(\neg
x_1 \lor x_2)\land (\neg x_2\lor x_3)\ldots (\neg x_{\ell}\lor x_1)$. Repeating
this for all variables that appear at least $4$ times produces an equivalent
instance $\phi'$ with $O(n+m)$ variables and clauses such that all literals
appear at most $2$ times. Furthermore, any satisfying assignment the formula forces exactly one true and one false literal in the new clauses.

We construct a digraph as follows: for each variable $x_i$ we make a digon and
label its vertices $x_i, \neg x_i$. We call this part of the digraph the
assignment part. For each clause we make a directed cycle of size equal to the
clause and associate each vertex of the cycle with a literal. We call this part
the satisfaction part. Finally, for each vertex of the assignment part we
connect it with digons with each vertex of the satisfaction part that
represents the opposite literal.

The digraph we constructed has maximum degree $6$; indeed, each literal has degree two in the assignment part and since each literal appears in at most two clauses, it has degree at most $4$ in the satisfaction part.
If there is a satisfying assignment then we give color $1$ to all True literals
of both parts and color $2$ to everything else. Observe that all arcs
connecting the two parts are bichromatic and if the assignment is satisfying
all directed cycles are also bichromatic.  For the converse direction, if there
is a $2$-coloring we can extract an assignment by setting to True all literals
which have color $1$ in the assignment part. Note that this implies that in the
satisfaction part all literals which have color $1$ have been set to True and
all literals which have color $2$ have been set to False, because of the digons
connecting the two parts. But this implies that our assignment is satisfying
because all cycles are bichromatic.  \end{proof}

\section{Bounded Feedback Vertex Set} \label{sec:bounded:fvs}

In this section we study the complexity of the problem parameterized by the
size of the feedback vertex set of a digraph. Throughout we will assume that a
feedback vertex set is given to us; if not we can use known FPT algorithms to
find the smallest such set \cite{ChenLLOR08}. 

As we are mentioned already, a digraph of DFVS $k-1$ can be always colored with $k$ colors. 
Our main result in this section is that $k$-\textsc{Digraph Coloring} is
NP-hard for digraphs of DFVS $k$. 
Observe that Remark~\ref{remark:dfvs:k-1:k:colorable} indicates 
that this result will be best possible.

\begin{remark} \label{remark_fvs=k_digons=kchoose2} Let $D=(V,E)$ be a digraph
with feedback vertex set $F$ of size $|F|=k$. If $F$ does not induce a
bi-directed clique, then $D$ is $k$-colorable.  \end{remark} 

Indeed, if $u,v\in F$ are not connected by a digon we can use one color for
$\{u,v\}$, $k-2$ distinct colors for the rest of $F$, and the remaining color
for the rest of the graph. Remark \ref{remark_fvs=k_digons=kchoose2} will also
be useful later in designing an algorithm, but at this point it is interesting
because it tells us that, since the graphs we construct in our reduction have
DFVS $k$ and must in some cases have $\vec{\chi}(D)>k$, our reduction needs to
construct a bi-directed clique of size $k$.

Before we go on to our reduction let us also mention that we will reduce from a
restricted version of \textsc{3-SAT} with the following properties: (i) all
clauses must have either only positive literals or only negative
literals (ii) all variables appear at most $2$ times positive and $1$ time
negative. We call this \textsc{Restricted-3-SAT}.

\begin{lemma}\label{lem:restricted} \textsc{Restricted-3-SAT} is NP-hard and
cannot be solved in $2^{o(n+m)}$ time unless the ETH is false. \end{lemma}

\begin{proof}
Start with an arbitrary instance
$\phi$ of \textsc{3-SAT} with $n$ variables and $m$ clauses. We first make sure
that every variable appears at most $3$ times as follows. 
First use the trick of Lemma \ref{thm:easy} to decrease the number of appearances of
each literal to two.  We now edit $\phi'$ as follows: for each variable
$x$ of $\phi'$ we replace every occurence of $\neg x$ with a fresh variable
$x'$. We then add the clause $(\neg x \lor \neg x')$. This gives a new
equivalent instance $\phi''$ which also has $O(n+m)$ variables and clauses and
satisfies all properties of \textsc{Restricted-3-SAT}.  \end{proof}

\begin{theorem} \label{theorem:2fvs:2col:NPhard}

For all $k\ge 2$, it is $NP$-hard to decide if a digraph $D=(V,E)$ is
$k$-colorable even when the size of its feedback vertex set is $k$.
Furthermore, this problem cannot be solved in time $2^{o(n)}$ unless the ETH is
false.

\end{theorem}

\begin{proof}

We will prove the theorem for $k=2$. To obtain the proof for larger values one
can add to the construction $k-2$ vertices which are connected to everything
with digons: this increases both the dichromatic number and the feedback vertex
set by $k-2$. Note that this does indeed construct a ``palette'' clique of size
$k$, as indicated by Remark \ref{remark_fvs=k_digons=kchoose2}.

We make a reduction from \textsc{Restricted-3-SAT}, which is NP-hard by Lemma
\ref{lem:restricted}. Our reduction will produce an instance of size linear in
the input formula, which leads to the ETH-based lower bound. Let $\phi$ be the
given formula with variables $x_1,\ldots, x_n$, and suppose that clauses
$c_1,\ldots, c_{\ell}$ contain only positive literals, while clauses
$c_{\ell+1},\ldots, c_m$ contain only negative literals. We will assume without
loss of generality that all variables appear in $\phi$ both positive and
negative (otherwise $\phi$ can be simplified).

We begin by constructing two ``palette'' vertices $v_1,v_2$ which are connected
by a digon. Then, for each clause $c_i, i\in [m]$ we do the following: if the
clause has size three we construct a directed path with vertices $l_{i,1},
w_{i,1}, l_{i,2}, w_{i,2}, l_{i,3}$, where the vertices $l_{i,1}, l_{i,2},
l_{i,3}$ represent the literals of the clause; if the clause has size two we
similarly construct a directed path with vertices $l_{i,1}, w_{i,1}, l_{i,2}$,
where again $l_{i,1}, l_{i,2}$ represent the literals of the clause.

For each variable $x_j, j\in [n]$ we do the following: for each clause
$c_{i_1}$ where $x_j$ appears positive and clause $c_{i_2}$ where $x_j$ appears
negative we construct a vertex $w'_{j,i_1,i_2}$ and add an incoming arc from
the vertex that represents the literal $x_j$ in the directed path of $c_{i_1}$
to $w'_{j,i_1,i_2}$; and an outgoing arc from $w'_{j,i_1,i_2}$ to the vertex
that represents the literal $\neg x_j$ in the directed path of $c_{i_2}$. 

Finally, to complete the construction we connect the palette vertices to the
rest of the graph as follows: $v_1$ is connected with a digon to all existing
vertices $w_{i,j}$, $i\in [m], j\in [2]$; $v_2$ is connected with a digon to
all existing vertices $w'_{j,i_1,i_2}$; $v_2$ has an outgoing arc to the first
vertex of each directed path representing a clause and an incoming arc from the
last vertex of each such path; $v_1$ has an outgoing arc to all vertices that
represent positive literals and an incoming arc from all vertices representing
negative literals.  (See Figure \ref{graph_sat})

\begin{figure}[h]
\centering
\begin{tikzpicture}[scale=0.9, transform shape]
\tikzstyle{vertex}=[circle, draw, inner sep=2pt,  minimum width=1 pt, minimum size=0.1cm]
\tikzstyle{vertex1}=[circle, draw, inner sep=2pt, fill=black!100, minimum width=1pt, minimum size=0.1cm]

\begin{scope}[xshift=-0.5cm]

\node[] () at (10,6) {$(\alpha)$};

\node[vertex1] (vf) at (11.5,9) {};
\node[] () at (8.5,9.3) {$v_1$};

\node[vertex1] (vt) at (8.5,9) {};
\node[] () at (11.5,9.3) {$v_2$};

\draw[-triangle 45] (vf) edge [bend left=10] (vt);
\draw[-triangle 45] (vt) edge [bend left=10] (vf);

\node[vertex1] (l1) at (8,7) {};
\node[vertex1] (u1) at (9,7) {};
\node[vertex1] (l2) at (10,7) {};
\node[vertex1] (u2) at (11,7) {};
\node[vertex1] (l3) at (12,7) {};

\node[] () at (8,6.6) {$l_{i,1}$};
\node[] () at (9,6.6) {$w_{i,1}$};
\node[] () at (10,6.6) {$l_{i,2}$};
\node[] () at (11,6.6) {$w_{i,2}$};
\node[] () at (12,6.6) {$l_{i,3}$};

\draw[-triangle 45] (l1)--(u1);
\draw[-triangle 45] (u1)--(l2);
\draw[-triangle 45] (l2)--(u2);
\draw[-triangle 45] (u2)--(l3);
\draw[-triangle 45] (l3)--(vf);
\draw[-triangle 45] (vf)--(l1);

\draw[-triangle 45] (u1) edge [bend left=10] (vt);
\draw[-triangle 45] (vt) edge [bend left=10] (u1);

\draw[-triangle 45] (u2) edge [bend left=10] (vt);
\draw[-triangle 45] (vt) edge [bend left=10] (u2);

\end{scope}

\begin{scope}[xshift=4cm]

\node[vertex1] (vf) at (10.7,9) {};
\node[] () at (9.3,9.3) {$v_1$};

\node[vertex1] (vt) at (9.3,9) {};
\node[] () at (10.7,9.3) {$v_2$};

\draw[-triangle 45] (vf) edge [bend left=10] (vt);
\draw[-triangle 45] (vt) edge [bend left=10] (vf);

\node[vertex1] (l1) at (9,7) {};
\node[vertex1] (l2) at (10,7) {};
\node[vertex1] (l3) at (11,7) {};

\node[] () at (8.8,6.6) {$l'=x_j$};
\node[] () at (10,6.6) {$w'_{j,i_1,i_2}$};
\node[] () at (11.3,6.6) {$l=\neg x_j$};

\node[] () at (10,6) {$(\beta)$};

\draw[-triangle 45] (vt)--(l1);
\draw[-triangle 45] (l1)--(l2);
\draw[-triangle 45] (l2)--(l3);
\draw[-triangle 45] (l3)--(vt);

\draw[-triangle 45] (l2) edge [bend left=10] (vf);
\draw[-triangle 45] (vf) edge [bend left=10] (l2);

\end{scope}

\begin{scope}[xshift=8.5cm]

\node[] () at (10,6) {$(\gamma)$};

\node[vertex1] (l1) at (8,9) {};
\node[vertex1] (u1) at (9,9) {};
\node[vertex1] (l2) at (10,9) {};
\node[vertex1] (u2) at (11,9) {};
\node[vertex1] (l3) at (12,9) {};

\node[] () at (8,9.3) {$x_{1}$};
\node[] () at (9,9.3) {$w_{1,1}$};
\node[] () at (10,9.3) {$x_{2}$};
\node[] () at (11,9.3) {$w_{1,2}$};
\node[] () at (12,9.3) {$x_{3}$};

\node[vertex1] (nl1) at (8.5,7) {};
\node[vertex1] (nu1) at (10,7) {};
\node[vertex1] (nl2) at (11.5,7) {};

\node[] () at (8.5,6.6) {$\neg x_1$};
\node[] () at (10,6.6) {$w_{2,1}$};
\node[] () at (11.5,6.6) {$\neg x_2$};

\draw[-triangle 45] (nl1)--(nu1);
\draw[-triangle 45] (nu1)--(nl2);

\draw[-triangle 45] (l1)--(u1);
\draw[-triangle 45] (u1)--(l2);
\draw[-triangle 45] (l2)--(u2);
\draw[-triangle 45] (u2)--(l3);

\node[vertex1] (u11) at (8.25,8) {};
\node[vertex1] (u12) at (10.75,8) {};

\node[] () at (8.8,8) {$w'_{1,1,2}$};
\node[] () at (11.3,8) {$w'_{2,1,2}$};

\draw[-triangle 45] (l1)--(u11);
\draw[-triangle 45] (u11)--(nl1);
\draw[-triangle 45] (l2)--(u12);
\draw[-triangle 45] (u12)--(nl2);

\end{scope}

\end{tikzpicture}
\caption{$(\alpha)$: The cycles created by $\{v_1,v_2\}$ and clauses with three literals. $(\beta)$: The cycles created by $\{v_1,v_2\}$ and each pair $\{x,\neg x\}$. $(\gamma)$: An example digraph for the formula $\phi = (x_1 \vee x_2 \vee x_3) \wedge (\neg x_1 \vee \neg x_2)$, without showing $v_1,v_2$.}
\label{graph_sat}
\end{figure}
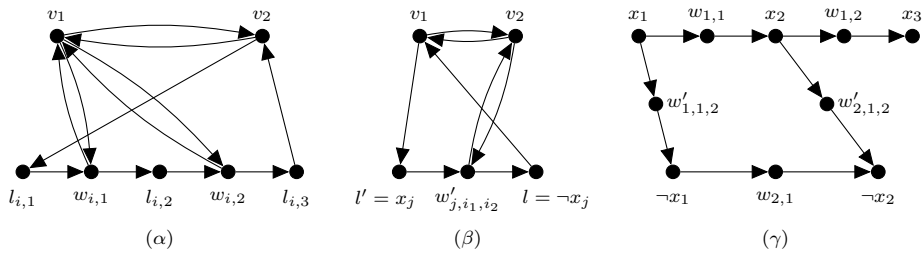\noindent

Let us now prove that this reduction implies the theorem. First, we claim that
in the digraph we constructed $\{v_1,v_2\}$ is a feedback vertex set. Indeed,
suppose we remove these two vertices. Now every arc in the remaining graph
either connects vertices that represent the same clause, or is incident on a
vertex $w'_{j,i_1,i_2}$. Observe that these vertices have only one incoming and
one outgoing arc and because of the ordering of the clauses $i_1< i_2$ (since
clauses that contain negative literals come later in the numbering). We
conclude that every directed path must either stay inside the path representing
the same clause or lead to a path the represents a later clause.  Hence, the
digraph is acyclic.

Let us now argue that if $\phi$ is satisfiable then the digraph is
$2$-colorable. We give color $1$ to $v_1$ and $2$ to $v_2$. We give color $2$
to each $w_{i,j}$ and color $1$ to each $w'_{j,i_1,i_2}$.  Fix a satisfying
assignment for $\phi$.  We give color $1$ to all  vertices $l_{i,j}$ that
represent literals set to True by the assignment and color $2$ to all remaining
vertices. Let us see why this coloring is acyclic. First, consider a vertex
$w'_{j,i_1,i_2}$. This vertex has color $1$ and one incoming and one outgoing
arc corresponding to opposite literals. Because the literals are opposite, one
of them has color $2$, hence $w'_{j,i_1,i_2}$ cannot be in any monochromatic
cycle and can be removed.  Now, suppose there is a monochromatic cycle of color
$1$. As $\{v_1,v_2\}$ is a feedback vertex set, this cycle must include $v_1$.
Since $v_2$ and all $w_{i,j}$ have color $2$ the vertex after $v_1$ in the
cycle must be some $l_{i,j}$ representing a positive literal which was set to
True by our assignment.  The only outgoing arc leaving from $l_{i,j}$ and going
to a vertex of color $1$ must lead it to a vertex $w'_{j',i,i'}$, which as we
said cannot be part of any cycle. Hence, no monochromatic cycle of color $1$
exists. Consider then a monochromatic cycle of color $2$, which must begin from
$v_2$. The next vertex on this cycle must be a $l_{i,1}$ and since we have
eliminated vertices $w'_{j,i_1,i_2}$ the cycle must continue in the directed
path of clause $i$. But, since we started with a satisfying assignment, at
least one of the literal vertices of this path has color $1$, meaning the cycle
cannot be monochromatic.

Finally, let us argue that if the digraph is $2$-colorable, then $\phi$ is
satisfiable. Consider a $2$-coloring which, without loss of generality, assigns
$1$ to $v_1$ and $2$ to $v_2$. The coloring must give color $2$ to all
$w_{i,j}$ and color $1$ to all $w_{j,i_1,i_2}$, because of the digons
connecting these vertices to the palette. Now, we obtain an assignment for
$\phi$ as follows: for each $x_j$, we find the vertex in our graph that
represents the literal $\neg x_j$ (this is unique since each variable appears
exactly once negatively): we assign $x_j$ to True if and only if this vertex has
color $2$.  Let us argue that this assignment satisfies all clauses. First,
consider a clause with all negative literals. If this clause is not satisfied,
then all the vertices representing its literals have color $2$. Because
vertices $w_{i,j}$ also all have color $2$, this creates a monochromatic cycle
with $v_2$, contradiction. Hence, all such clauses are satisfied. Second,
consider a clause $c_i$ with all positive literals. In the directed path
representing $c_i$ at least one literal vertex must have color $1$, otherwise
we would get a monochromatic cycle with $v_2$. Suppose this vertex represents
the literal $x_j$ and has an out-neighbor $w'_{j,i,i_2}$, which is colored $1$.
If the out-neighbor of $w'_{j,i_1,i_2}$ is also colored $1$, we get a
monochromatic cycle with $v_1$.  Therefore, that vertex, which represents the
literal $\neg x_j$ has color $2$.  But then, according to our assignment $x_j$
is True and $c_i$ is satisfied.  \end{proof}

The last result of this section concerns $2$-coloring of oriented graphs.

\begin{theorem} \label{theorem:oriented:2col:3fvs:NPhard}
It is $NP$-hard to decide if an oriented graph $D=(V,E)$ is $2$-colorable even when the size of its feedback vertex set is $3$.
\end{theorem}
\begin{proof}
We adapt the proof of Theorem~\ref{theorem:2fvs:2col:NPhard}. First let us give an intuition behind the gadget we are going to use. 
In the proof of Theorem~\ref{theorem:2fvs:2col:NPhard} the digraph we created is not an oriented graph as it contains digons. All the digons of that digraph are connected to vertices $v_1$ or $v_2$, and therefore, we want to replace $v_1$ and $v_2$ with a gadget that contains two arcs $t_1t_2$ and $f_1f_2$ such that the vertices $t_1$ and $t_2$ to have the same color as $v_1$ and the vertices $f_1$ and $f_2$ to have the same color as $v_2$.   
Then we can replace all cycles that contained $v_1$ (respectively, $v_2$) with cycles that contain the arc $t_1t_2$  (respectively, $f_1f_2$) and the rest of the proof will remain the same.

The gadget we use in place of  $\{v_1,v_2\}$ is the one in the Fig.~\ref{graph_red_oriented_sat}.
Furthermore, we will not use the digon between $v_1$ and $v_2$ and we replace all the other incoming arcs of $v_1$ from the previous construction with incoming arcs to $t_1$, the outgoing arcs of $v_1$ with outgoing arcs from $t_2$, the incoming arcs of $v_2$ with incoming arcs to $f_1$, the outgoing arcs of $v_2$ with outgoing arcs from $f_2$.
For example, the digon $v_1w_{i,1}$ in the gadget $(\alpha)$ from the previous theorem becomes a triangle  $t_1 t_2 w_{i,1}$.

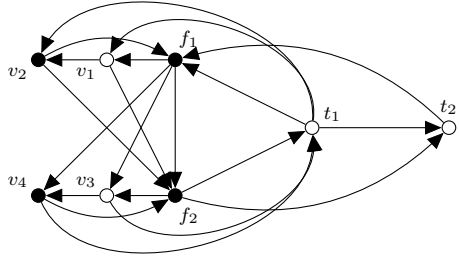
\begin{figure}[h]
\centering
\begin{tikzpicture}[scale=0.9, transform shape]
\tikzstyle{vertex}=[circle, draw, inner sep=2pt,  minimum width=1 pt, minimum size=0.1cm]
\tikzstyle{vertex1}=[circle, draw, inner sep=2pt, fill=black!100, minimum width=1pt, minimum size=0.1cm]

\node[vertex] (vt1) at (8,10) {};
\node[vertex] (vt2) at (10,10) {};

\node[] () at (8.3,10.2) {$t_1$};
\node[] () at (10,10.3) {$t_2$};

\node[vertex1] (vf1) at (6,9) {};
\node[vertex1] (vf2) at (6,11) {};

\node[] () at (6.2,8.7) {$f_2$};
\node[] () at (6.2,11.3) {$f_1$};

\node[vertex] (cf21) at (5,11) {};
\node[vertex1] (cf22) at (4,11) {};
\node[] () at (4.7,10.8) {$v_1$};
\node[] () at (3.7,10.8) {$v_2$};

\node[vertex] (cf11) at (5,9) {};
\node[vertex1] (cf12) at (4,9) {};
\node[] () at (4.7,9.2) {$v_3$};
\node[] () at (3.7,9.2) {$v_4$};

\draw[-triangle 45] (vt1)--(vt2);
\draw[-triangle 45] (vf2)--(vf1);
\draw[-triangle 45] (vf1)--(vt1);
\draw[-triangle 45] (vt1)--(vf2);

\draw[-triangle 45] (vf1) edge [bend right] (vt2);
\draw[-triangle 45] (vt2) edge [bend right] (vf2);

\draw[-triangle 45] (vf1)--(cf11);
\draw[-triangle 45] (cf11)--(cf12);
\draw[-triangle 45] (cf12) edge [bend right] (vf1);

\draw[-triangle 45] (vf2)--(cf21);
\draw[-triangle 45] (cf21)--(cf22);
\draw[-triangle 45] (cf22) edge [bend left] (vf2);

\draw[-triangle 45] (cf21)--(vf1);
\draw[-triangle 45] (cf22)--(vf1);
\draw[-triangle 45] (vt1) edge [bend right=80] (cf21);
\draw[-triangle 45] (vt1) edge [bend right=80] (cf22);

\draw[-triangle 45] (vf2)--(cf11);
\draw[-triangle 45] (vf2)--(cf12);
\draw[-triangle 45] (cf11) edge [bend right=80] (vt1);
\draw[-triangle 45] (cf12) edge [bend right=80] (vt1);

\node[] () at (7,8) {};

\end{tikzpicture}
\caption{Gadget $H$: It is $2$-colorable, and in any $2$-coloring both pairs $\{f_1,f_2\}$ and $\{t_1,t_2\}$ must be monochromatic but with different color per pair.
}\label{graph_red_oriented_sat}
\end{figure}\noindent

Now we need to show that in any proper $2$-coloring of this gadget 
both pairs ${f_1,f_2}$ and ${t_1,t_2}$ are monochromatic and we use different color per pair.

First observe that there exists such a coloring (see Fig.~\ref{graph_red_oriented_sat})
We will show that the vertices $f_1$ and $t_1$ cannot have the same color. Assume that they are both colored $0$; then the vertex $f_2$ must be colored $1$ because we have the cycle $f_1,f_2,t_1$. Because the vertex $f_2$ is colored $1$ and there exists the cycle $f_2,v_3,v_4$ we know that at least one of $v_3$, $v_4$ must be colored $0$. Let $v_3$ (respectively, $v_4$) be colored $0$, then the coloring is not proper because there exists the cycle $t_1,f_1,v_3$ (resp. $t_1,f_1,v_4$) with all the vertices colored $0$. This is a contradiction so the $f_1$ and $t_1$ cannot have the same color. Similarly we can prove that $f_2$ and $t_1$ cannot have the same color. So we must color the vertices $f_1$ and $f_2$ with one color and $t_1$ with the second. Furthermore because we have the cycle $f_1,f_2,t_2$, the vertex $t_2$ must use the same color as $t_1$.

It remains to show that the size of minimum feedback vertex set is at most $3$; observe that the set $\{f_1,f_2,t_1\}$ is a feedback vertex set (see Fig.~\ref{feedback_vertex_set}).

\begin{figure}[H]
\centering
\begin{tikzpicture}[scale=0.9, transform shape]
\tikzstyle{vertex}=[circle, draw, inner sep=2pt,  minimum width=1 pt, minimum size=0.1cm]
\tikzstyle{vertex1}=[circle, draw, inner sep=2pt, fill=black!100, minimum width=1pt, minimum size=0.1cm]

\node[vertex] (u1) at (10,10) {};
\node[] () at (10,10.3) {$t_2$};
\node[vertex] (u2) at (9,11) {};
\node[vertex1] (u3) at (8,11) {};
\node[] () at (8.7,10.8) {$v_1$};
\node[] () at (7.7,10.8) {$v_2$};

\node[vertex] (u4) at (9,9) {};
\node[vertex1] (u5) at (8,9) {};
\node[] () at (8.7,9.2) {$v_3$};
\node[] () at (7.7,9.2) {$v_4$};

\draw[-triangle 45] (u4)--(u5);
\draw[-triangle 45] (u2)--(u3);

\draw (12,11) -- (13,11);
\draw (12,9) -- (13,9);
\draw (12,11) -- (12,9);
\draw (13,11) -- (13,9);
\draw[-triangle 45] (12,10)--(u1);
\node[] () at (9,8) {$H$ without $\{f_1,f_2,t_1\}$};
\node[] () at (12.5,8) {remaining digraph};

\end{tikzpicture}
\caption{For the remaining digraph, it has been proved that is acyclic in the previous theorem so $\{f_1,f_2,t_1\}$ is a feedback vertex set of the whole digraph.}\label{feedback_vertex_set}
\end{figure}
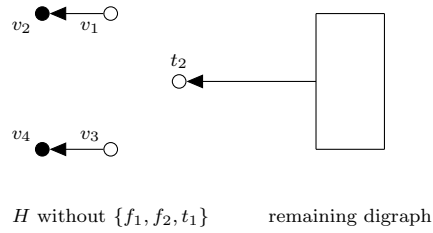\noindent

\end{proof}

This result is tight as, by Remark~\ref{remark_fvs=k_digons=kchoose2}, we know that oriented graphs with DFVS of size $k$ are $k$-colorable.

\section{Bounded Feedback Arc Set and Bounded Degree} \label{sec:bounded:fas:degree}

In this section we first present two algorithmic results: we show that
$k$-\textsc{Digraph Coloring} becomes FPT (by $k$) if either the input graph
has feedback vertex set $k$ and maximum degree at most $4k-3$; or if it has
feedback arc set at most $k^2-1$ (and unbounded degree). Interestingly, the
latter of these results is exactly tight and the former is almost tight: in the
second part we refine the reduction of the previous section to show that
$k$-\textsc{Digraph Coloring} is NP-hard for digraphs which have simlutaneously
a FAS of size $k^2$, a feedback vertex set of size $k$ and maximum  degree
$\Delta=4k-1$.

\subsection{Algorithmic Results}

Our first result shows that for $k$-\textsc{Digraph Coloring}, if we are
promised a feedback vertex set of size $k$ (which is the smallest value for
which the problem is non-trivial), then the problem remains tractable for
degree up to $4k-3$. Observe that in the case of general digraphs (where we
do not bound the feedback vertex set) the problem is already hard for maximum
degree $2k+2$ (see Other Related Work section), so this seems encouraging. However, we show in Theorem
\ref{theorem:fvs=k:degree=4k-1:NPhard} that this tractability cannot be
extended much further.

\begin{theorem}\label{thm:algfvs} Let $D=(V,E)$ be a digraph with feedback
vertex set $F$ of size $|F|=k$ and maximum degree $\Delta \le 4k-3$. Then, $D$
is $k$-colorable if and only if $D[N[F]]$ is $k$-colorable. Furthermore, a
$k$-coloring of $D[N[F]]$ can be extended to a $k$-coloring of $D$ in
polynomial time.  \end{theorem}

\begin{proof}

Let $D=(V,E)$ be such a digraph. If $D[N[F]]$ is not $k$-colorable, then $D$ is
not $k$-colorable, so we need to prove that if $D[N[F]]$ is $k$-colorable then
$D$ is $k$-colorable and we can extend this coloring to $D$.  Assume that
$D[N[F]]$ is $k$-colorable. By Remark~\ref{remark_fvs=k_digons=kchoose2} we can
assume that $D[F]$ is a bi-directed clique. Let $c:N[F] \to [k]$ be the assumed
$k$-coloring and without loss of generality say that $F=\{v_1,\ldots, v_k\}$
and $c(v_i)=i$ for all $i\in [k]$.

Before we continue let us define the following sets of vertices:  we will call
$V_{i,in}$ the set of vertices $v \in N[F] \setminus F$ such that $c(v)=i$ and
there exists an arc $vv_i \in E$. Similarly we will call $V_{i,out}$ the set of
vertices $v \in N[F] \setminus F$ where $c(v)=i$ and there exists an arc $v_iv
\in E$. The sets $V_{i,in}$ and $V_{i,out}$ are disjoint in any proper coloring
(otherwise we would have a monochromatic digon). Furthermore, $V_{i,in}\cup
V_{i,out}$ is disjoint from $V_{j,in}\cup V_{j,out}$ for $j\neq i$ (because
their vertices have different colors), so all these $2k$ sets are pair-wise
disjoint.  We first show that if one of these $2k$ sets is empty, then we can
color $D$.

\begin{claim} \label{claim:i:j:neighborhood} If for some $i \in [k]$ one of the
sets $V_{i,in}$, $V_{i,out}$ is empty then we can extend $c$ to a $k$-coloring
of $D$ in polynomial time.  \end{claim} 

\begin{proof} We keep $c$ unchanged and color all of $V(D)\setminus N[F]$ with
color $i$. This is a proper $k$-coloring. Indeed, this cannot create a
monochromatic cycle with color $j\neq i$. Furthermore, if a monochromatic cycle
of color $i$ exists, since this cycle must intersect $F$, we conclude that it
must contain $v_i$. However, in the current $k$-coloring $v_i$ either has
in-degree or out-degree $0$ in the vertices colored $i$, so no monochromatic
cycle can go through it.  \end{proof}

In the remainder we assume that all sets $V_{i,in}, V_{i,out}$ are non-empty.
Our strategy will be to edit the $k$-coloring of $D[N[F]]$ so that we retain a
proper $k$-coloring, but one of these $2k$ sets becomes empty. We will then
invoke Claim \ref{claim:i:j:neighborhood} to complete the proof.

We now define, for each pair $i,j \in[k] $ with $i<j$ the set $E_{i,j}$ which
contains all arcs with one endpoint in $\{v_i,v_j\}$ and the other in
$V_{i,in}\cup V_{i,out}\cup V_{j,in}\cup V_{j,out}$ and whose endpoints have
\emph{distinct} colors. We call $E_{i,j}$ the set of \emph{cross} arcs for the
pair $(i,j)$. We will now argue that for some pair $(i,j)$ we must have
$|E_{i,j}|\le 3$. For the sake of contradiction, assume that $|E_{i,j}|\ge 4$
for all pairs. Then, by summing up the degrees of vertices of $F$ we have:

$$ \sum_{i \in [k]} d(v_i) \geq 2k + k(2k-2) + \sum_{i,j\in [k], i<j} |E_{i,j}|
\geq 2k^2 + 4{k\choose 2} = 4k^2-2k$$

In the first inequality we used the fact that each $v_i\in F$ has at least two
arcs connecting it to $V_{i,in}\cup V_{i,out}$ (since these sets are
non-empty); $2k-2$ arcs connecting it to other vertices of $F$ (since $F$ is a
clique); and each cross arc of a set $E_{i,j}$ contributes one to the degree of
one vertex of $F$. From this calculation we infer that the average degree of
$F$ is at least $4k-2$, which is a contradiction, since we assumed that the
digraph has maximum degre $4k-3$.

Fix now $i,j$ such that $|E_{i,j}|\le 3$. We will recolor $V_{i,in}\cup
V_{i,out}\cup V_{j,in} \cup V_{j,out}$ in a way that allows us to invoke Claim
\ref{claim:i:j:neighborhood}. Since we do not change any other color, we will
only need to prove that our recoloring does not create monochromatic cycles of
colors $i$ or $j$ in $D[N[F]]$. We can assume that $|E_{i,j}|=3$, since if
$|E_{i,j}|<3$ we can add an arbitrary missing cross arc and this can only make
the recoloring process harder. Furthermore, without loss of generality, we
assume that $v_i$ has strictly more cross arcs of $E_{i,j}$ incident to it than
$v_j$.

We now have to make a case analysis. First, suppose all three arcs of $E_{i,j}$
are incident on $v_i$. Then, there exists a set among $V_{j,in}, V_{j,out}$
that has at most one arc connecting it to $v_i$. We color this set $i$, and
leave the other set colored $j$. We also color $V_{i,in}\cup V_{i,out}$ with
$j$.  This creates no monochromatic cycle because: (i) $v_i$ now has at most
one neighbor colored $i$ in $V_{i,in}\cup V_{i,out}\cup V_{j,in} \cup
V_{j,out}$, so no monochromatic cycle goes through $v_i$; (ii) $v_j$ has either
no out-neighbors or no in-neighbors colored $j$ in $V_{i,in}\cup V_{i,out}\cup
V_{j,in} \cup V_{j,out}$.  With the new coloring we can invoke Claim
\ref{claim:i:j:neighborhood}. In the remainder we therefore assume that two
arcs of $E_{i,j}$ are incident on $v_i$ and one is incident on $v_j$.

Second, suppose that one of $V_{j,in}, V_{j,out}$ has no arcs connecting it to
$v_i$. We color this set $i$ and leave the other set colored $j$. Observe that
one of $V_{i,in}, V_{i,out}$ has no arc connecting it to $v_j$. We color that
set $j$ and leave the other set colored $i$. In the new coloring both $v_i,
v_j$ either have no out-neighbor or no in-neighbor with the same color in
$V_{i,in}\cup V_{i,out}\cup V_{j,in}\cup V_{j,out}$, so the coloring is proper
and we can invoke Claim \ref{claim:i:j:neighborhood}. In the remainder we
assume that $v_i$ has one arc connecting it to each of $V_{j,in}, V_{j,out}$.

Third, suppose that both arcs of $E_{i,j}$ incident on $v_i$ have the same
direction (into or out of $v_i$). We then color $V_{i,in}\cup V_{i,out}$ with
$j$ and $V_{j,in}\cup V_{j,out}$ with $i$. In the new coloring $v_j$ has at
most one neighbor with the same color and $v_i$ has either only in-neighbors or
only out-neighbors with color $i$, so the coloring is acyclic and we again
invoke Claim \ref{claim:i:j:neighborhood}.

Finally, we have the case where two arcs of $E_{i,j}$ are incident on $v_i$,
they have different directions, one has its other endpoint in $V_{j,in}$ and
the other in $V_{j,out}$.  Observe that one of $V_{i,in}, V_{i,out}$ has no arc
connecting it to $v_j$ and suppose without loss of generality that it is
$V_{i,in}$ (the other case is symmetric).  We color $V_{i,in}$ with $j$ and
leave $V_{i,out}$ with color $i$. One of $V_{j,in}, V_{j,out}$ has an incoming
arc from $v_i$; we color this set $i$ and leave the other colored $j$. Now,
$v_i$ only has out-neighbors with color $i$, while $v_j$ has at either only
in-neighbors or only out-neighbors colored $j$, so we are done in this final
case.  \end{proof}

Our second result concerns a parameter more restricted than feedback vertex
set, namely feedback arc set. Note that, in a sense, the class of graphs of
bounded feedback arc set contains the class of graphs that have bounded
feedback vertex set and bounded degree (selecting all incoming or outgoing arcs
of each vertex of a feedback vertex set produces a feedback arc set), so the
following theorem may seem more general.  However, a closer look reveals that
the following result is incomparable to Theorem \ref{thm:algfvs}, because
graphs of feedback vertex set $k$ and maximum degree $4k-3$ could have feedback
arc set of size up to almost $2k^2$ (consider for example a bi-direction of the
complete graph $K_{k,2k-2}$), while the following theorem is able to handle
graphs of unbounded degree but feedback arc set up to (only) $k^2-1$. As we
show in Theorem \ref{theorem:fvs=k:degree=4k-1:NPhard}, this is tight.

\begin{theorem} \label{theorem:k-col:fas<k^2:poly} Let $D$ be a digraph with a
feedback arc set $F$ of size at most $k^2-1$. Then $D$ is $k$-colorable if and
only if $D[V(F)]$ is $k$-colorable, and such a coloring can be extended to $D$
in polynomial time. \end{theorem}

\begin{proof} 

It is obvious that if $D[V(F)]$ is not $k$-colorable then $D$ is not
$k$-colorable.  We will prove the converse by induction.  For $k=1$ it is trivial
to see that if $|F|=0$ then $D$ is acyclic so is $1$-colorable. Assume then
that any digraph $D$ with a feedback arc set $F$ of size at most $(k-1)^2-1$ is
$(k-1)$-colorable, if and only if $D[V(F)]$ is $(k-1)$-colorable.

Suppose now that we have $D$ with a feedback arc set $F$ with $|F|\le k^2-1$
and we find that $D[V(F)]$ is $k$-colorable (this can be tested in $2^{O(k^2)}$
time). Let $c:V(F)\to [k]$ be a coloring of $V(F)$.  We distinguish two cases:

\textit{Case 1.} There exists a color class (say $V_k$) such that at least
$2k-1$ arcs of $F$ are incident on $V_k$. Then $D - V_k$ has a feedback
arc set of size at most $|F| -(2k-1) \le k^2-1-(2k-1) \le (k-1)^2-1$ and
$V_1,\ldots,V_{k-1}$ remains a valid coloring of the remainder of $V(F)$. So by
inductive hypothesis $D - V_k$ has a $(k-1)$-coloring.  By using the
color $k$ on $V_k$ we have a $k$-coloring for $D$.

\textit{Case 2.} Each color class is incident on at most $2k-2$ arcs of F. We
then claim that there is a way to color $V(F)$ so that all arcs of $F$ have
distinct colors on their endpoints. If we achieve this, we can trivially extend
the coloring to the rest of the graph, as arcs of $F$ become irrelevant.  Now,
let us call $v\in V(F)$ as \textit{type one} if $v$ is incident on at least $k$
arcs of $F$. We will show that there is at most one \textit{type one} vertex in
each color class. Indeed,  if $u,v\in V_i$ are both \textit{type one}, then
they are incident on at least $2k-1$ arcs of $F$ (there is no digon between $u$
and $v$ because they share a color), which we assumed is not the case, as $V_i$
is incident on at most $2k-2$ arcs of $F$.  Therefore, we can use $k$ distinct
colors to color all the \textit{type one} vertices of $V(F)$. Each remaining
vertex of $V(F)$ is incident on at most $k-1$ arcs of $F$. We consider these
vertices in some arbitrary order, and give each a color that does not already
appear on the other endpoints of its incident arcs from $F$.  Such a color
always exists, and proceeding this way we color all arcs of $F$ with distinct
colors.  This completes the proof. \end{proof}

\subsection{Hardness}

In this section we improve upon our previous reduction by producing a graph
which has bounded degree and bounded feedback arc set. Our goal is to do this
efficiently enough to (almost) match the algorithmic bounds given in the
previous section. 

\begin{theorem}\label{theorem:fvs=k:degree=4k-1:NPhard} For all $k\ge 2$, it is
$NP$-hard to decide if a digraph $D=(V,E)$ is $k$-colorable, even if $D$ has a
feedback vertex set of size $k$, a feedback arc set of size $k^2$, and maximum
degree $\Delta=4k-1$.  \end{theorem}

\begin{proof}

Recall that in the proof of Theorem \ref{theorem:2fvs:2col:NPhard} for $k\ge 2$
we construct a graph that is made up of two parts: the palette part, which is a
bi-directed clique that contains $v_1, v_2$ and the $k-2$ vertices we have
possibly added to increase the chromatic number (call them $v_3,\ldots, v_k$);
and the part that represents the formula. We perform the same reduction except
that we will now edit the graph to reduce its degree and its feedback arc set.
In particular, we delete the palette vertices and replace them with a gadget
that we describe below. 

We construct a new palette that will be a bi-directed clique of size $k$, whose
vertices are now labeled $v^i, i\in[k]$. Let $M$ be the number of vertices of
the graph we constructed for Theorem \ref{theorem:2fvs:2col:NPhard}. We
construct $M$ ``rows'' of $2k$ vertices each. More precisely, for each $\ell
\in [M], i\in [k]$ we construct the two vertices $v^i_{\ell,in}, v^i_{\ell,
out}$.  In the remainder, when we refer to row $\ell$, we mean the set $\{
v^i_{\ell, in}, v^i_{\ell,out}\ |\ i\in [k]\}$. For all $i,j\in [k], i<j$ we
connect the vertices of row $1$ to the vertices of the clique as shown in
Figure \ref{fig:graph_H_0}. For all $i,j\in [k], i<j$ and $\ell \in [M-1]$ we
connect the vertices of rows $\ell, \ell+1$ as shown in Figure
\ref{fig:graph_H_i+1}.

In more detail we have:

\begin{enumerate}

\item For each $i\in [k]$ the vertex $v^i$ has an arc to all $v^j_{1,out}$ for
$j\ge i$, an arc to $v^j_{1,in}$ for all $j\neq i$, and an arc from
$v^j_{1,in}$ for all $j\le i$. 

\item For each $\ell \in [M]$, for all $i<j$ we have the following four arcs:
$v^j_{\ell,out}v^i_{\ell,out}$, $v^i_{\ell,out}v^j_{\ell,in}$,
$v^j_{\ell,in}v^i_{\ell,in}$, and $v^j_{\ell,out}v^i_{\ell,in}$. As a result,
inside a row arcs are oriented from $out$ to $in$ vertices; and between
vertices of the same type from larger to smaller indices $i$.

\item For each $\ell \in [M-1]$, for all $i\in [k]$ we have the arcs
$v^i_{\ell,out}v^i_{\ell+1,out}$ and $v^i_{\ell+1,in}v^i_{\ell,in}$. As a
result, the $v^i_{\ell,out}$ vertices form a directed path going out of $v^i$
and the $v^i_{\ell,in}$ vertices form a directed path going into $v^i$.

\item For each $\ell \in [M-1]$, for all $i,j \in [k]$ with $i<j$ we have the
arcs $v^i_{\ell,out}v^j_{\ell+1,in}$, $v^i_{\ell,out}v^j_{\ell+1,out}$,
$v^i_{\ell+1,in}v^j_{\ell,in}$, $v^j_{\ell,out}v^i_{\ell+1,in}$. Again, arcs
incident on an $out$ and an $in$ vertex are oriented towards the $in$ vertex.

\end{enumerate}

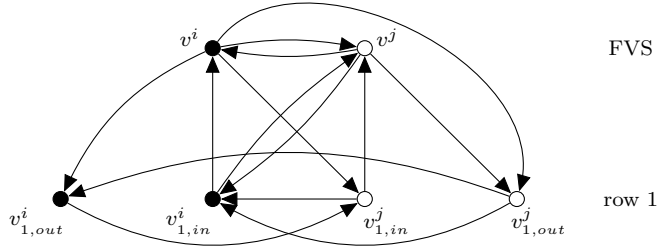
\begin{figure}[ht] \centering \begin{tikzpicture}[scale=1.0, transform shape]
\tikzstyle{vertex}=[circle, draw, inner sep=2pt,  minimum width=1 pt, minimum
size=0.1cm] \tikzstyle{vertex1}=[circle, draw, inner sep=2pt, fill=black!100,
minimum width=1pt, minimum size=0.1cm]

\begin{scope}[xshift=6cm]

\node[vertex1] (v) at (4,11) {};
\node[vertex] (u) at (6,11) {};

\node[vertex1] (v1) at (4,9) {};
\node[vertex] (u1) at (6,9) {};

\node[vertex1] (v2) at (2,9) {};
\node[vertex] (u2) at (8,9) {};

\node[] () at (9.5,9) {row 1};
\node[] () at (9.5,11) {FVS};

\node[] () at (3.7,11.15) {$v^i$};
\node[] () at (6.3,11.15) {$v^j$};
\node[] () at (3.7,8.7) {$v^i_{_{1,in}}$};
\node[] () at (6.3,8.7) {$v^j_{_{1,in}}$};
\node[] () at (1.7,8.7) {$v^i_{_{1,out}}$};
\node[] () at (8.3,8.7) {$v^j_{_{1,out}}$};

\draw[-triangle 45] (v) edge [bend left=10] (u);
\draw[-triangle 45] (u) edge [bend left=10] (v);

\draw[-triangle 45] (v) -- (u1);

\draw[-triangle 45] (v1) -- (v);

\draw[-triangle 45] (u1) -- (u);

\draw[-triangle 45] (v2) edge [bend right=30] (u1);

\draw[-triangle 45] (u2) edge [bend left=30] (v1);

\draw[-triangle 45] (u1) -- (v1);


\draw[-triangle 45] (v) edge [bend left=80] (u2);

\draw[-triangle 45] (u2) edge [bend right=20] (v2);

\draw[-triangle 45] (v1) edge [bend left=10] (u);
\draw[-triangle 45] (u) edge [bend left=10] (v1);

\draw[-triangle 45] (v)  edge [bend right=20] (v2);
\draw[-triangle 45] (u) -- (u2);

\end{scope}

\end{tikzpicture}
\caption{Graph showing the connections between two vertices of the clique palette ($v^i, v^j$, where $i<j$) and the corresponding vertices of row $1$.} \label{fig:graph_H_0} \end{figure}

\begin{figure}[ht]
\centering
\begin{tikzpicture}[scale=1.0, transform shape]
\tikzstyle{vertex}=[circle, draw, inner sep=2pt,  minimum width=1 pt, minimum size=0.1cm]
\tikzstyle{vertex1}=[circle, draw, inner sep=2pt, fill=black!100, minimum width=1pt, minimum size=0.1cm]

\begin{scope}[xshift=6cm]

\node[vertex1] (v1) at (4,9) {};
\node[vertex] (u1) at (6,9) {};

\node[vertex] (u2) at (8,9) {};
\node[vertex1] (v2) at (2,9) {};

\node[] () at (9.5,9.2) {row $\ell$};
\node[] () at (9.5,8.8) {(for $\ell\geq 1$)};

\node[] () at (9.5,7) {row $\ell+1$};

\node[vertex] (u11) at (6,7) {};
\node[vertex] (u12) at (8,7) {};
\node[vertex1] (v12) at (2,7) {};
\node[vertex1] (v11) at (4,7) {};

\node[] () at (3.6,9.1) {$v^i_{_{\ell,in}}$};
\node[] () at (6.4,9.1) {$v^j_{_{\ell,in}}$};
\node[] () at (1.5,9.1) {$v^i_{_{\ell,out}}$};
\node[] () at (8.4,9.2) {$v^j_{_{\ell,out}}$};

\node[] () at (8.3,6.6) {$v^j_{_{\ell+1,out}}$};
\node[] () at (6.3,6.7) {$v^j_{_{\ell+1,in}}$};
\node[] () at (3.8,6.7) {$v^i_{_{\ell+1,in}}$};
\node[] () at (1.7,6.7) {$v^i_{_{\ell+1,out}}$};



\draw[-triangle 45] (v11) -- (v1);

\draw[-triangle 45] (v2) -- (v12);

\draw[-triangle 45] (u11) -- (u1);

\draw[-triangle 45] (u2) -- (u12);

\draw[-triangle 45] (u11) -- (v11);
\draw[-triangle 45] (u2) -- (v11);
\draw[-triangle 45] (v11) -- (u1);
\draw[-triangle 45] (u12) edge [bend left=30] (v11);
\draw[-triangle 45] (v2) -- (u12);
\draw[-triangle 45] (v2) -- (u11);


\draw[-triangle 45] (u12) edge [bend right=40] (v12);
\draw[-triangle 45] (v12) edge [bend right=30] (u11);
\end{scope}

\end{tikzpicture}
\caption{Here we present the way we are connecting the vertices of the rows $i$ and $i+1$}
\label{fig:graph_H_i+1}
\end{figure}
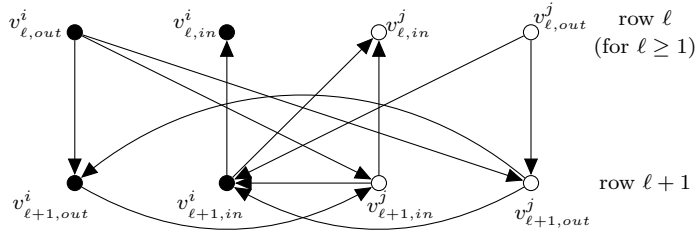

Let $P$ be the gadget we have constructed so far, consisting of the clique of
size $k$ and the $M$ rows of $2k$ vertices each. We will establish the
following properties.

\begin{enumerate}

\item Deleting all vertices $v^i, i\in [k]$ makes $P$ acyclic and eliminates
all directed paths from any vertex $v^i_{\ell,in}$ to any vertex
$v^j_{\ell',out}$, for all $i,j\in [k], \ell,\ell'\in [M]$.

\item The maximum degree of any vertex of $P$ is $4k-2$. 

\item There is a $k$-coloring of $P$ that gives all vertices of $\{
v^i_{\ell,in}, v^i_{\ell,out}\ |\ \ell\in [M]\}$ color $i$, for all $i\in [k]$.

\item In any $k$-coloring of $P$, for all $i$, all vertices of $\{
v^i_{\ell,in}, v^i_{\ell,out}\ |\ \ell\in [M]\}$ receive the same color as
$v^i$.

\end{enumerate}

Before we go on to prove these four properties of $P$, let us explain why they
imply the theorem. To complete the construction, we insert $P$ in our graph in
the place of the palette clique we were previously using. To each vertex of the
original graph, we associate a distinct row of $P$ (there are sufficiently many
rows to do this). Now, if vertex $u$ of the original graph, which is associated
to row $\ell$, had an arc from (respectively to) the vertex $v_i$ in the
palette, we add an arc from $v^i_{\ell,out}$ (respectively to $v^i_{\ell,in}$).

Let us first establish that the new graph has the properties promised in the
theorem. The maximum degree of any vertex in the main (non-palette) part
remains unchanged and is $2k+2\le 4k-1$ while the maximum degree of any vertex
of $P$ is now at most $4k-1$, as we added at most one arc to each vertex.
Deleting $\{ v^i\ |\ i\in [k]\}$ eliminates all cycles in $P$, but also all
cycles going through $P$, because such a cycle would need to use a path from a
vertex $v^i_{\ell, in}$ (since these are the only vertices with incoming arcs
from outside $P$) to a vertex $v^j_{\ell',out}$. Deleting all of $P$ leaves the
graph acyclic (recall that the palette clique was a feedback vertex set in our
previous construction), so there is a feedback vertex set of size $k$. 

For the feedback arc set we remove the arcs $\{ v^jv^i\ |\ j>i,\ i,j\in [k] \}
\cup \{ v^i_{1,in}v^j\ |\ j> i,\ i,j\in [k]\}\cup \{v^i_{1,in}v^i\ |\ i\in
[k]\} $. Note that these are indeed $k^2$ arcs. To see that this is a feedback
arc set, suppose that the graph contains a directed cycle after its removal.
This cycle must contain some vertex $v^i$, because we argued that $\{ v^i\ |\
i\in [k]\}$ is a feedback vertex set. Among these vertices, select the $v^i$
with minimum $i$. We now examine the arc of the cycle going into $v^i$. Its
tail cannot be $v^j$ for $j>i$, as we have removed such arcs, nor $v_j$ for
$j<i$, as this contradicts the minimality of $i$. It cannot be $v^i_{1,in}$ as
we have also removed these arcs. And it cannot be $v^j_{1,in}$ for $j<i$, as
these arcs are also removed. But no other in-neighbor of $v^i$ remains,
contradiction.

Let us also argue that using the gadget $P$ instead of the palette clique does
not affect the $k$-colorability of the graph. This is not hard to see because,
following Properties 3 and 4 we can assume that any $k$-coloring of $P$ will
give color $i$ to all vertices of $\{v^i\}\cup \{ v^i_{\ell,in},
v^i_{\ell,out}\ |\ \ell\in [M] \}$. The important observation is now that, for
all $\ell\in [M]$ there will always exist a monochromatic path from $v^i$ to
$v^i_{\ell,out}$ and from $v^i_{\ell,in}$ to $v^i$. We now note that, if we fix
a coloring of the non-palette part of the graph, this coloring contains a
monochromatic cycle involving vertex $v_i$ of the original palette if and only
if the same coloring gives a monochromatic cycle in the new graph going through
$v^i$. 

Finally, we need to prove the four properties. 

\noindent
\textbf{Property 1.} Once we delete $\{ v^i\ |\ i\in [k]\}$ we observe
that for every vertex $v^i_{\ell,in}$ its only outgoing arcs are to vertices
$v^j_{\ell,in}$ for $j<i$ or vertices $v^j_{\ell-1,in}$ for $j\ge i$. This
shows that we have eliminated all directed paths from $v^i_{\ell,in}$ to
$v^j_{\ell',out}$. Furthermore, this shows that no cycle can be formed using
$v^i_{\ell,in}$ vertices, since all their outgoing arcs either move to a
previous row, or stay in the same row but decrease $i$. In a similar way, no
directed cycle can be formed using only $v^i_{\ell,out}$ vertices, as all their
outgoing arcs either move to a later row, or stay in the same row but decrease
$i$. 

\noindent
\textbf{Property 2.} For a vertex $v^i$ we have $2k-2$ arcs incident on
it from the clique; the two arcs connecting it to $v^i_{1,in}, v^i_{1,out}$;
two arcs connecting it to $v^j_{1,in}, v^j_{1,out}$ for $j>i$; two arcs
connecting it to $v^j_{1,in}$ for $j<i$. This gives $2k-2 + 2i + 2(k-i) =
4k-2$. 

For a vertex $v^i_{1,in}$ we have one arc to $v^j$ for $j\le i$; two arcs to
$v^j$ for $j>i$; arcs to all $v^j_{in}, v^j_{out}$ for $j\neq i$; arcs to
$v^j_{2,in}$ for $j\le i$. This gives $i + 2(k-i) + 2(k-1) + i = 4k-2 $.

For a vertex $v^i_{1,out}$ we have arcs from $v^j$ for $j\le i$; arcs to
$v^j_{1,in}, v^j_{1,out}$ for $j\neq i$; arcs to $v^j_{2,in}$ and $v^j_{2,out}$
for $j\ge i$; arcs to  $v^j_{2,in}$ for all $j< i$. This gives $i + 2(k-1) +
2(k-i) + i = 4k-2$.

For a vertex $v^i_{\ell,in}$, $\ell\ge 2$ we have arcs to $v^j_{\ell-1,in}$,
for $j\ge i$; to $v^j_{\ell-1,out}$ for $j\neq i$; to $v^j_{\ell,in},
v^j_{\ell,out}$ for $j\neq i$; from $v^j_{\ell+1,in}$ for $j\le i$. This gives
$(k-i+1) + (k-1) + 2(k-2) + i = 4k-2$.

Finally, for a vertex $v^i_{\ell,out}$, $\ell \ge 2$ we have arcs from
$v^j_{\ell-1,out}$ for $j\le i$; to $v^j_{\ell,in}, v^j_{\ell,out}$ for $j\neq
i$; to all $v^j_{\ell+1,in}$, for $j\neq i$; to $v^j_{\ell+1,out}$ for $j\ge
i$.  This gives $i + 2(k-1) + (k-1) + (k-i+1) = 4k-2$. 

\noindent
\textbf{Property 3.}  We assign color $i$ to $v^i$ and to $\{
v^i_{\ell,in}, v^i_{\ell,out}\ |\ \ell\in [M] \}$. We claim that there is no
monochromatic cycle in $P$ with this coloring. Indeed, if such a cycle exists,
it must use $v^i$, as $\{ v^i\ |\ i\in [k]\}$ is a feedback vertex set. But
observe that with the coloring we gave, for each $\ell\in [M-1]$ the only
out-neighbor of $v^i_{\ell,out}$ with color $i$ is $v^i_{\ell+1,out}$ and
$v^i_{M,out}$ has no out-neighbor of color $i$. Similar examination of $\{
v^i_{\ell,in}\ |\ \ell\in [M]\}$ shows that the part of $P$ colored $i$ induces
a directed path on $2M+1$ vertices with $v^i$ in the middle.

\noindent
\textbf{Property 4.} Since the vertices $v^i$ induce a clique, we may
assume without loss of generality that we are given a coloring $c$ where
$c(v^i)=i$.  We prove the property by induction on $\ell$. For $\ell=1$, we
will first prove that $c(v^i_{1,in})=i$ by induction on $i$. For the base case
we have that $v^1_{1,in}$ is connected with a digon with $v^j$ for all $j>1$,
so $c(v^1_{1,in})=1$. Now, fix a $j$ and suppose that for all $i<j$ we have
$c(v^i_{1,in})=i$. Then $v^j_{1,in}$ cannot receive any color $i<j$, because
this would make a cycle with $v^i_{1,in}, v^i$. It can also not receive a color
$i>j$ because it has a digon to all $v^i$ for $i>j$. Hence, $c(v^j_{1,in})=j$.
Continuing on $\ell=1$, we will prove by reverse induction on $i$ that
$c(v^i_{1,out})=i$. For $c(v^k_{1,out})$ if we give this vertex any color $j<k$
then we get a cycle with $v^j, v^j_{1,in}$, so we must have $c(v^k_{1,out})=k$.
Now fix an $i$ and suppose that for all $j>i$ we have $c(v^j_{1,out})=j$. If we
give $v^i_{1,out}$ a color $j>i$ this will make a cycle with $v_j, v^j_{1,out},
v^i_{1,out}, v^j_{1,in}$.  But if we give $v^i_{1,out}$ a smaller color $j<i$,
this will also make a cycle with $v^j, v^j_{1,in}$. Therefore,
$c(v^i_{1,out})=i$ for all $i$.

Suppose now that the property is true for row $\ell$ and we want to prove it
for row $\ell+1$. We will use similar reasoning as in the previous case. We
will also use the observation that for all $i$, there is a monochromatic path
from $v^i$ to $v^i_{\ell,out}$ and a monochromatic path from $v^i_{\ell,in}$ to
$v^i$.  First, we show by induction on $i$ that $c(v^i_{\ell+1,in})=i$ for all
$i$.  For $v^1_{\ell+1,in}$ we observe that if we give this vertex color $j>1$,
then using the arcs from $v^j_{\ell,out}$ and to $v^j_{\ell,in}$ we have a
monochromatic cycle of color $j$. Hence, $c(v^1_{\ell+1,in})=1$. Fix a $j$ and
suppose that for all $i<j$ we have $c(v^i_{\ell+1,in})=i$. If we assign
$c(v^j_{\ell+1,in})$ a color $i<j$, then we get a cycle using $v^i_{\ell,out},
v^j_{\ell+1,in}, v^i_{\ell+1,in}, v^i_{\ell,in}$. If we assign it a color
$i>j$, then we get the cycle using $v^i_{\ell,out}, v^j_{\ell+1,in},
v^i_{\ell,in}$. So, for all $i$ we have $c(v^i_{\ell+1,in})=i$. To complete the
proof, we do reverse induction to show that $c(v^i_{\ell+1,out})=i$. For
$c(v^k_{\ell+1,out})$ we cannot give this vertex color $j<k$ because this will
give a cycle using $v^j_{\ell, out}, v^k_{\ell+1, out}, v^j_{\ell+1,in}
v^j_{ell,in}$. Now, fix an $i$ and assume that for $j>i$ we have
$c(v^j_{\ell+1,out})=j$. We cannot assign $v^i_{\ell+1,out}$ any color $j>i$
because this would give the cycle $v^j_{\ell,out}, v^j_{\ell+1,out},
v^i_{\ell+1,out}, v^j_{\ell+1, in}, v^j_{\ell, in}$. We can also not assign any
color $j<i$ as this gives the cycle using $v^j_{\ell,out}, v^i_{\ell+1,out},
v^j_{\ell+1, in}, v^j_{\ell, in}$. We conclude that for all $i$ we have
$c(v^i_{\ell+1,out})=i$.

\end{proof}

\section{Treewidth} \label{sec:bounded:treewidth}

In this section we consider the complexity of \textsc{Digraph Coloring} with
respect to parameters measuring the acyclicity of the underlying graph, namely,
treewidth and treedepth. Before we proceed let us recall that in all graphs $G$
we have $\chi(G) \le \tw(G)+1 \le \td(G)+1$. This means that if our goal is
simply to obtain an FPT algorithm then parameterizing by treewidth implies that
the graph's chromatic number (and therefore also the digraph's dichromatic
number) is bounded.  We first present an algorithm with complexity
$k^{\tw}(\tw!)$ which, using the above argument, proves that \textsc{Digraph
Coloring} is FPT parameterized by treewidth.

\begin{theorem}\label{thm:tw} There is an algorithm which, given a digraph $D$
on $n$ vertices and a tree decomposition of its underlying graph of width $\tw$
decides if $D$ is $k$-colorable in time $k^{\tw}(\tw!)n^{O(1)}$. \end{theorem}

\begin{proof}
The proof uses standard techniques so we sketch some details. In particular we
assume that we are given a nice tree decomposition on which we will perform
dynamic programming.  Before we proceed, let us slightly recast the problem. We
will say that a digraph $D=(V,E)$ is $k$-colorable if there exist two functions
$c,\sigma$ such that (i) $c:V\to [k]$ partitions $V$ into $k$ sets (ii)
$\sigma$ is an ordering of $V$ (iii) for all arcs $uv\in E$ we have either
$c(u)\neq c(v)$ or $\sigma(u)<\sigma(v)$. It is not hard to see that this
reformulation is equivalent to the original problem. Indeed, if we have a
$k$-coloring, since each color class is acyclic, we can find a topological
ordering $\sigma_i$ of the graph $G[V_i]$ induced by each color class and then
concatenate them to obtain an ordering of $V$. For the converse direction, the
existence of $c,\sigma$ implies that if we look at vertices of each color
class, $\sigma$ must induce a topological ordering, hence each color class is
acyclic.

Now, let $D$ be a digraph and $S$ be a subset of its vertices. Let $(c,\sigma)$
be a pair of coloring and ordering functions that prove that $D$ is
$k$-colorable. Then, we will say that the signature of solution $(c,\sigma)$
for set $S$ is the pair $(c_S,\sigma_S)$ where $c_S: S\to [ k ]$ is defined as
$c_S(u) = c(u)$ and $\sigma_S: S\to [ |S| ]$ is an ordering function such that
for all $u,v\in S$ we have $\sigma_S(u)<\sigma_S(v)$ if and only if
$\sigma(u)<\sigma(v)$. In other words, the signature of a solution is the
restriction of the solution to the set $S$.

Given a rooted nice tree decomposition of $D$, let $B_t$ be a bag of the
decomposition and denote by $B_t^{\downarrow}$ the set of vertices of $D$ which
are contained in $B_t$ and bags in the sub-tree rooted at $B_t$. Our dynamic
programming algorithm stores for each $B_t$ a collection of all pairs
$(c,\sigma)$ such that there exists a $k$-coloring of $D[B_t^{\downarrow}]$
whose signature is $(c,\sigma)$. If we manage to construct such a table for
each node, it will suffice to check if the collection of signatures of the root
is empty to decide if the graph is $k$-colorable.

The table is easy to initialize for Leaf nodes, as the only valid signature
contains the empty coloring and ordering function. For an Introduce node that
adds $u$ to a bag containing $B_t$ we consider all signatures $(c,\sigma)$ of
contained in the table of the child bag. For each such signature we construct a
signature $(c', \sigma')$ which is consistent with $(c,\sigma)$ but also colors
$u$ and places it somewhere in the ordering (we consider all such
possibilities). For each $(c',\sigma')$ we delete this signature if $u$ has a
neighbor in the bag who is assigned the same color by $c'$ but such that their
arc violates the topological ordering $\sigma'$. We keep all other produced
signatures. To see that this is correct observe that $u$ has no neighbors in
$B_t^{\downarrow}\setminus B_t$, because all bags are separators, so if we
produce an ordering of $B_t^{\downarrow}$ consistent with $\sigma'$ the only
arcs incident on $u$ that could violate it are contained in the bag (and have
been checked). For Forget nodes the table is easily update by keeping only the
restrictions of valid signatures to the new bag. Finally, for Join nodes we
keep a signature $(c,\sigma)$ if and only if it is valid for both sub-trees.
Again this is correct because nodes of one sub-tree not contained in the bag do
not have neighbors in the other sub-tree, so as long as we produce an ordering
consistent with $\sigma$ we can concatenate we cannot violate the topological
ordering condition.

For the running time observe that the size of the DP table is $k^{\tw}(\tw!)$,
because we consider all colorings and all ordering of each bag. In Introduce
nodes we spend polynomial time for each entry of the child node (checking all
placements of the new vertex), while computation in Join nodes can be performed
in time linear in the size of the table. So the running time is in the end
$k^{\tw}(\tw!)n^{O(1)}$.  \end{proof}

As we explained, even though Theorem \ref{thm:tw} implies that \textsc{Digraph
Coloring} is FPT parameterized by treewidth, the complexity it gives is
significantly worse than the complexity of \textsc{Coloring}, which is
essentially $k^{\tw}$. Our main result in this section is to show that this is
likely to be inevitable, even if we focus on the more restricted case of
treedepth and $2$ colors.

\begin{theorem} If there exists an algorithm which decides if a given digraph
on $n$ vertices and (undirected) treedepth $\td$ is $2$-colorable in time
$\td^{o(\td)}n^{O(1)}$, then the ETH is false. \end{theorem}

\begin{proof}

Suppose we are given a \textsc{3-SAT} formula $\phi$ with $n$ variables and $m$
clauses. We will produce a digraph $G$ such that $|V(G)|=2^{O(n/\log n)}m$ and
$\td(G)=O(n/\log n)$ and $G$ is $2$-colorable if and only if $\phi$ is
satisfiable. Before we proceed, observe that if we can construct such a graph
the theorem follows, as an algorithm with running time $O^*(\td^{o(td)})$ for
$2$-coloring $G$ would decide the satisfiability of $\phi$ in time $2^{o(n)}$.

To simplify presentation we assume without loss of generality that $n$ is a
power of $2$ (otherwise adding dummy variables to $\phi$ can achieve this while
increasing $n$ be a factor of at most $2$). We begin the construction of $G$ by
creating $\log n$ independent sets $V_1,\ldots,V_{\log n}$, each of size
$\lceil \frac{2en}{\log^2n}\rceil$. We add a vertex $u$ and connect it with
arcs in both directions to all vertices of $\cup_{i\in[\log n]} V_i$. We also
partition the variables of $\phi$ into $\log n$ sets $X_1,\ldots,X_{\log n}$ of
size at most $\lceil \frac{n}{\log n}\rceil$. 

The main idea of our construction is that the vertices of $V_i$ will represent
an assignment to the variables of $X_i$. Observe that all vertices of $V_i$ are
forced to obtain the same color (as all are forced to have a distinct color
from $u$), therefore the way these vertices represent an assignment is via
their topological ordering in the DAG they induce together with other vertices
of the graph which obtain the same color.

To continue our construction, for each $i\in[\log n]$ we do the following: we
enumerate all the possible truth assignments of the variables of $X_i$ and for
each such truth assignment $\sigma:X_i\to \{0,1\}^{|X_i|}$ we define (in an
arbitrary way) a distinct ordering $\rho(\sigma)$ of the vertices of $V_i$. We
will say that the ordering $\rho(\sigma)$ is the \emph{translation} of
assignment $\sigma$. Note that there are $|V_i|! \ge (\frac{2en}{\log^2n})! \ge
(\frac{2n}{\log^2n})^{\frac{2en}{\log^2n}} = 2^{\frac{2en}{\log^2n}(1+\log
n-2\log\log n)}> 2^{\lceil\frac{n}{\log n}\rceil}$ for $n$ sufficiently large,
so it is possible to translate truth assignments to $X_i$ to orderings of $V_i$
injectively. Note that enumerating all assignments for each group takes time
$2^{O(n/\log n)}  = 2^{o(n)}$.

Consider now a clause $c_j$ of $\phi$ and suppose some variable of the group
$X_i$ appears in $c_j$. For each truth assignment $\sigma$ to $X_i$ which
satisfies $c_j$ we construct an independent set $S_{j,i,\sigma}$ of size
$|X_i|-1$, label its vertices $s_{j,i,\sigma}^{\ell}$, for $\ell\in [|X_i|-1]$.
For each $\ell$ we add an arc from $\rho(\sigma)^{-1}(\ell)$ to
$s_{j,i,\sigma}^{\ell}$ and an arc from $s_{j,i,\sigma}^{\ell}$ to
$\rho(\sigma)^{-1}(\ell+1)$. In other words, the $\ell$-th vertex of
$S_{j,i,\sigma}$ has an incoming arc from the vertex of $V_i$ which is
$\ell$-th according to the ordering $\rho(\sigma)$ which is the translation of
assignment $\sigma$ and an outgoing arc to the vertex of $V_i$ which is in
position $(\ell+1)$ in the same ordering. Observe that this implies that if all
vertices of $V_i$ and of $S_{j,i,\sigma}$ are given the same color, then the
topological ordering of the induced DAG will agree with $\rho(\sigma)$ on the
vertices of $V_i$. 

To complete the construction, for each clause $c_j$ we do the following: take
all independent sets $S_{j,i,\sigma}$ which we have constructed for $c_j$ and
order them in a cycle in some arbitrary way. For two sets $S_{j,i,\sigma},
S_{j,i',\sigma'}$ which are consecutive in this cycle add a new ``connector''
vertex $p_{j,i,\sigma,i',\sigma'}$, all arcs from $S_{j,i,\sigma}$ to this
vertex, and all arcs from this vertex to $S_{j,i',\sigma'}$. Finally, we
connect each connector vertex $p_{j,i,\sigma,i',\sigma'}$ we have constructed
to an arbitrary vertex of $V_1$ with a digon. This completes the construction.

Let us argue that if $\phi$ is satisfiable, then $G$ is $2$-colorable. We color
$u$ with color $2$, all the vertices in $V_i$ for $i\in[\log n]$ with $1$ and
all connector vertices $p_{i,j,\sigma,i',\sigma'}$ with $2$.  For each clause
$c_j$ there exists a group $X_i$ that contains a variable of $c_j$ such that
the supposed satisfying assignment of $\phi$, when restricted to $X_i$ gives an
assignment $\sigma:X_i\to\{0,1\}^{|X_i|}$ which satisfies $c_j$.  Therefore,
there exists a corresponding set $S_{j,i,\sigma}$. Color all vertices of this
set with $1$. After doing this for all clauses, we color all other vertices
with $2$.  We claim this is a valid $2$-coloring.  Indeed, the graph induced by
color $2$ is acyclic, as it contains $u$ (but none of its neighbors) and for
each $c_j$, all but one of the sets $S_{j,i,\sigma}$ and the vertices
$p_{j,i,\sigma,i',\sigma'}$.  Since these sets have been connected in a
directed cycle throught connector vertices, and for each $c_j$ we have colored
one of these sets with $1$, the remaining sets induce a DAG. For the graph
induced by color $1$ consider for each $V_i$ the ordering $\rho(\sigma)$, where
$\sigma$ is the satisfying assignment restricted to $V_i$. Every vertex outside
$V_i$ which received color $1$ and has arcs to $V_i$, has exactly one incoming
and one outgoing arc to $V_i$.  Furthermore, the directions of these arcs agree
with the ordering $\rho(\sigma)$. Hence, since $\cup_{i\in[\log n]}V_i$ touches
all arcs with both endpoints having color $1$ and all such arcs respect the
orderings of $V_i$, the graph induced by color $1$ is acyclic. 

For the converse direction, suppose we have a $2$-coloring of $G$. Without loss
of generality, $u$ has color $2$ and $\cup_{i\in[\log n]} V_i$ has color $1$.
Furthermore, all connectors $p_{j,i,\sigma,i',\sigma'}$ also have color $2$.
Consider now a clause $c_j$.  We claim that there must be a group
$S_{j,i,\sigma}$ such that $S_{j,i,\sigma}$ does not use color $2$.  Indeed, if
all such groups use color $2$, since they are linked in a directed cycle with
all possible arcs between consecutive groups and connectors, color $2$ would
not induce a DAG.  So, for each $c_j$ we find a group $S_{j,i,\sigma}$ that is
fully colored $1$ and infer from this the truth assignment $\sigma$ for the
group $X_i$. Doing this for all clauses gives us an assignment that satisfies
every clause.  However, we need to argue that the assignment we extract is
consistent, that is, there do not exist $S_{j,i,\sigma}$ and $S_{j',i,\sigma'}$
which are fully colored $1$ with $\sigma\neq \sigma'$. For the sake of
contradiction, suppose that two such sets exist, and recall that
$\rho(\sigma)\neq \rho(\sigma')$. We now observe that if $S_{j,i,\sigma}\cup
V_i$ only uses color $1$, then any topological ordering of $V_i$ in the graph
induced by color $1$ must agree with $\rho(\sigma)$, which is a total ordering
of $V_i$. In a similar way, the ordering of $V_i$ must agree with
$\rho(\sigma')$, so if $\sigma\neq \sigma'$ we get a contradiction.

Finally, let us argue about the parameters of $G$. For each clause $c_j$ of
$\phi$ we construct an independent set of size $O(n/\log^2n)$ for each
satisfying assignment of a group $X_i$ containing a variable of $c_j$. There
are at most $3$ such groups, and each group has at most $2^{n/\log n}$
satisfying assignments for $c_j$, so $|V(G)|=2^{O(n/\log n)}m$. 

For the treedepth, recall that deleting a vertex decreases treedepth by at most
$1$. We delete $u$ and all of $\cup_{i\in[\log n]}V_i$ which are $O(n/\log n)$
vertices in total. It now suffices to prove that in the remainder all
components have treedepth $O(n/\log n)$.  In the remainder every component is
made up of the directed cycle formed by sets $S_{j,i,\sigma}$ and connectors
$p_{j,i,\sigma,i',\sigma'}$. We first delete a vertex
$p_{j,i,\sigma,i',\sigma'}$ to turn the cycle into a directed ``path'' of
length $L=2^{O(n/\log n)}$. We now use the standard argument which proves that
paths of length $L$ have treedepth $\log L$, namely, we delete the
$p_{j,i,\sigma,i',\sigma'}$ vertex that is closest to the middle of the path
and then recursively do the same in each component.  This shows that the
remaining graph has treedepth logarithmic in the length of the path, therefore
at most $O(n/\log n)$.  \end{proof}

\section{2-Coloring Tournaments} \label{sec:2-col:tournament}

In this section we propose an algorithm that decides if a given tournament $T$ is $2$-colorable in time $O^*(\sqrt[3]{6}^n)$. Our algorithm starts by removing, arbitrarily, as many disjoint triangles from the tournament as possible and then considers all the proper partial colorings of the tournament induced on these triangles. Then we use a recursive algorithm in order to determine if any of these partial colorings can be extended to a proper $2$-coloring for the whole tournament.

\begin{algorithm}[h]
\caption{[$2$-COL($T$) decision function]}
\label{Alg1}
\begin{algorithmic}[1]
  \Require A tournament $T=(V,E)$.
  \Ensure Is $\overrightarrow{\chi}(T)=2$ or not?
  \State $V_1 \leftarrow \emptyset$, $V_2 \leftarrow V$
  \State $IsTwoDC\leftarrow \textbf{False}$
  \While{there is a triangle $\{v_1,v_2,v_3\}$ in $V_2$}
  \State $V_1 \leftarrow V_1 \cup \{v_1,v_2,v_3\}$ \State $V_2 \leftarrow V_2 \setminus \{v_1,v_2,v_3\}$
  \EndWhile
  \For{all 2-coloring $\mathcal{C}:V_1\rightarrow \{1,2\}$ that are proper}
  \State $IsTwoDC \leftarrow$ $IsTwoDC$ $\vee$ Ext $2$-DCN($T$,$V_1$,$\mathcal{C}$) \label{alg1:calls:alg2}
  \EndFor
  \State \Return $IsTwoDC$
  \end{algorithmic}
\end{algorithm}  
\noindent

As we mentioned, the previous algorithm uses another one in order to decide if a partial coloring is extendable.

In order to decide it, we search for two types of triangles in the tournament - triangles that contain one uncolored vertex and two vertices with the same color and triangles that contain only one colored vertex. For the first type, it is easy to see that we know which color we have to assign to the uncolored vertex. However, for the second type, the algorithm calls itself in order to decide if any of the possible colorings is extendable to this triangle.

\begin{algorithm}[h]
\caption{[Ext $2$-COL($T$,$V_C$,$\mathcal{C}$ ) decision function]}
\label{Alg2}
\begin{algorithmic}[1]
  \Require A tournament $T=(V,E)$, a set of vertices $V_C \subseteq V $ and a function $\mathcal{C}:V_C\rightarrow \{1,2\}$.
  \Ensure Can we find a proper $2$-coloring for $T$ by extending $\mathcal{C}$?
  \State $V_{NC} \leftarrow V \setminus V_C$
  \State $Ext\leftarrow \textbf{False}$
  \While{there is a triangle $\{v_1,v_2,v_3\}$ such that $v_1 \in V_{NC}$, $v_2,v_3 \in V_C$ and $\mathcal{C}(v_2)=\mathcal{C}(v_3)$ }
  \State $V_C \leftarrow V_C \cup \{v_1\}$,  $V_{NC} \leftarrow V_{NC} \setminus \{v_1\}$
  \State set $\mathcal{C}(v_1)$ to be the color that is not $\mathcal{C}(v_2)$ = $\mathcal{C}(v_3)$ \label{ext:triangles:type1}
  \EndWhile	
  \If{$\mathcal{C}$ is a proper coloring for $V_C$}
  \While{there is a triangle $\{v_1,v_2,v_2\}$ such that $v_1,v_2 \in V_{NC}$ and $v_3 \in V_C$}
  \State $V_C \leftarrow V_C \cup \{v_1,v_2\}$,  $V_{NC} \leftarrow V_{NC} \setminus \{v_1,v_2\}$
  \ForAll{the pairs $\{Col_1,Col_2\} \neq \{\mathcal{C}(v_3),\mathcal{C}(v_3)\}$}  \label{ext:triangles:type2:start}
  \State set $\mathcal{C}(v_1) \leftarrow Col_1$ and $\mathcal{C}(v_2) \leftarrow Col_2$
  \State $Ext \leftarrow Ext$ $\vee$ Ext $2$-DNC($T$,$V_C$,$\mathcal{C}$) \label{line:is:extendable} 
  \EndFor \label{ext:triangles:type2:end}
  \EndWhile
  \EndIf
  \State for all $v \in V_{NC}$ set $\mathcal{C}(v)$ to be 1
  \If{$\mathcal{C}$ is a proper coloring for $V$}
  \State $Ext \leftarrow \textbf{True}$ \label{line:proper:coloring}
  \EndIf
  \State \Return $Ext$ 
  \end{algorithmic}
\end{algorithm}  

Before we continue to the proof, let us recall that any tournament $T$ that has a directed cycle must contain a triangle.
Therefore, in the Algorithm~\ref{Alg1} we know that the graph $T[V\setminus V_1]$ where $V_1$ is the set that we use in the line~\ref{alg1:calls:alg2}, is acyclic as we could not find any other triangles in it.

Now, let us prove that the Algorithm~\ref{Alg2} does what we claim.

\begin{lemma} \label{lemma_ext_coloring}
Given a tournament $T=(V,E)$, a set of vertices $S \subseteq V$ such that $T[V\setminus S]$ is acyclic and a function $\mathcal{C} :S \rightarrow \{1,2\} $ of $S$, Algorithm~\ref{Alg2} applied to $V_C=S$ decides if we can find a function $\mathcal{C}^* :V\rightarrow \{1,2\} $ that gives a proper $2$-coloring for the tournament $T$ such that $\mathcal{C}^*(v)=\mathcal{C}(v)$ for all $v\in S$.
\end{lemma}

\begin{proof}

If the function cannot be extended the algorithm will return \textbf{False} because in order to change the value to \textbf{True} that means that at one of the calls of the algorithm  we checked an extension $\mathcal{C}^*$ of $\mathcal{C}$ and it was a proper coloring for the tournament which is a contradiction.
So we have to prove that if the given function $\mathcal{C}$ can be extended in order to give a proper coloring of the whole tournament then the algorithm will return \textbf{True}. For the rest of the proof let us call the triangles that contain one uncolored vertex and two vertices of the same color as \textit{type one} and the  triangles with two uncolored vertices as \textit{type two}. 
Let $\mathcal{C}$ be extendable (i.e., there is an extension $\mathcal{C}^*$ that gives a proper coloring for the tournament); the algorithm first checks if there exists a triangle of \textit{type one} and gives to the uncolored vertices the other color (in line~\ref{ext:triangles:type1}). It is clear that this is the only option for these vertices so that the new color function remains extendable. After that the algorithm checks for triangles of \textit{type two}. In this case we know that the two uncolored vertices cannot have both the same color as the third vertex; so we have a total of $2^2-1=3$ cases. After that the algorithm checks (between lines~\ref{ext:triangles:type2:start} and \ref{ext:triangles:type2:end}) if any of these possibilities can be can be extended and gives us a proper coloring (by calling itself in line~\ref{line:is:extendable}). 

As we mentioned, the algorithm tries to extend all the possible colorings (except those that are not proper) so at some point we have an extendable function $\mathcal{C}$ and either we do not have any uncolored vertices or we do not have any triangles of \textit{type one} or \textit{two}. 

\textit{Case 1.} Suppose $V_{NC} = \emptyset$ when line 16 of Algorithm 2 is executed. Then $\mathcal{C}$ is a proper coloring of $V$ which means that after the check in line~\ref{line:proper:coloring} we change the value of the variable $Ext$ to \textbf{True}. 

It remains to show that in the second case if we colored the remaining uncolored vertices with any color we have a proper coloring for $T$.

\textit{Case 2.} In this case we do not have any triangles of \textit{type one} or \textit{two}. This combined with the assumption that the coloring is extendable implies that by coloring the remaining vertices with any color we end up with a coloring that does not have any monochromatic triangle. It remains to show the following claim:

\begin{claim}
Let $T=(V,E)$ be a tournament and $C:V\rightarrow \{1,2\}$ a function that is a $2$-coloring of $T$ such that there is no monochromatic triangle. Then $C$ is a proper coloring.
\end{claim}

\begin{proof}
Assume that $C$ does not give a proper coloring. Then there must exist a monochromatic cycle $S$ with length grater than $3$. Note that $S$ induces a tournament. But any tournament which contains a directed cycle contains a triangle. This gives a contradiction since there are no monochromatic triangles in $T$.  
 \end{proof}

So our coloring is a proper $2$-coloring; thus the algorithm will change the value of the variable $Ext$ to \textbf{True} in line~\ref{line:proper:coloring} and due to the logic \textit{or} in line~\ref{line:is:extendable} this \textbf{True} will be kept until the algorithm terminates.
 \end{proof}

Finally we are going to prove that Algorithm~\ref{Alg1} decides if a tournament is $2$-colorable and that it runs in $O^*(\sqrt[3]{6}^n)$ time.

\begin{theorem} \label{lemma:alg:2col:turnament}
Given a tournament $T=(V,E)$, Algorithm~\ref{Alg1} decides if $T$ is $2$-colorable.
\end{theorem}
\begin{proof}
It is easy to see that Algorithm~\ref{Alg1} tries to extend any proper coloring of $V_1$. Now, in order to use Lemma~\ref{lemma_ext_coloring} we need to observe that we have no triangles in $V_2$. Since a tournament without triangles is acyclic, it follows that $V_2$ is acyclic. So, from lemma~\ref{lemma_ext_coloring} we know that if one of these colorings can be extended then the Algorithm~\ref{Alg2} will return \textbf{True}. Thus, Algorithm~\ref{Alg1} returns \textbf{True} if the tournament is $2$-colorable and \textbf{False} otherwise.
\end{proof}

\begin{theorem} \label{theorem:2col:tournamnt}
Let $T=(V,E)$ be a tournament. Then we can decide if the dichromatic number of $T$ is two in time $\mathcal{O}^*(\sqrt[3]{6}^n)$
\end{theorem}
\begin{proof}
Observe that in Algorithm~\ref{Alg1} all the steps are polynomial except the number of the proper ways to color the set $V_1$ and the time Algorithm~\ref{Alg2} needs. Now, it is easy to see that the number of proper ways to color $V_1$ is at most $6^{\frac{|V_1|}{3}}$ since for every triangle in $V_1$ we know that we have six possible choices to color it (all except the two that give to every vertex the same color). This means that we call Algorithm~\ref{Alg2} at most $6^{\frac{|V_1|}{3}}$ times. The running time of the second algorithm depends on the number of times that it will call itself. Now we can see that for the remaining vertices ($V_2= V\setminus V_1$), in the worst case, we need to check three different colorings (see proof of lemma~\ref{lemma_ext_coloring}) for two vertices at a time. Thus, the running time of Algorithm~\ref{Alg2} is $\mathcal{O}^*(3^{\frac{|V_2|}{2}})$. So, we can decide if $T$ is $2$-colorable in time 
$$ \mathcal{O}^*(6^{\frac{|V_1|}{3}} \cdot 3^{\frac{|V_2|}{2}}) = \mathcal{O}^* (\sqrt[3]{6}^{|V_1|+|V_2|})=\mathcal{O}^* (\sqrt[3]{6}^{n})$$
 \end{proof}

\section{Conclusions}

In this paper we have strengthened known results about the complexity of
\textsc{Digraph Coloring} on digraphs which are close to being DAGs, precisely
mapping the threshold of tractability for DFVS and FAS; and we precisely
bounded the complexity of the problem parameterized by treewidth, uncovering an
important discrepancy with its undirected counterpart. One question for further
study is to settle the degree bound for which $k$-\textsc{Digraph Coloring} is
NP-hard for DFVS $k$, and more generally to map out how  the tractability
threshold for the degree evolves for larger values of the DFVS  from
$4k-\Theta(1)$ to $2k+\Theta(1)$, which is the correct threshold when the DFVS
is unbounded.  With regards to undirected structural parameters, it would be
interesting to investigate whether a $\mathrm{vc}^{o(\mathrm{vc})}$ algorithm
exists for $2$-\textsc{Digraph Coloring}, where $\mathrm{vc}$ is the input
graph's vertex cover, as it seems challenging to extend our hardness result to
this more restricted case.



\bibliography{ref}

\end{document}